\documentclass[journal]{IEEEtran}
\IEEEoverridecommandlockouts

\usepackage{cite}
\usepackage{amsmath,amssymb,amsfonts}
\usepackage{epsfig,subfigure,multirow}
\usepackage{algorithmic}
\usepackage{graphicx}
\usepackage{textcomp}
\usepackage{xcolor, color}
\usepackage[pagebackref=false]{hyperref}
\usepackage{array}
\usepackage{pstricks}
\usepackage{verbatim}
\usepackage{stfloats}
\usepackage[bookmarks=false]{}
\usepackage{aecompl}
\usepackage{mathrsfs}
\usepackage{arydshln}
\usepackage{mathtools} 
\usepackage{amsthm}
\usepackage{caption}
\usepackage{lipsum}
\usepackage{bm}
\usepackage{graphicx}
\usepackage{algorithm}
\usepackage{algorithmic}
\usepackage{comment}

\def \BibTeX{{\rm B\kern-.05em{\sc i\kern-.025em b}\kern-.08em
		T\kern-.1667em\lower.7ex\hbox{E}\kern-.125emX}}

\newtheorem{theorem}{Theorem}

\newtheorem{lemma}{Lemma}
\newtheorem{example}{Example}

\newtheorem{remark}{Remark}

\begin{document}
	
	\title{Fundamental Limits of Multi-Message Private Computation}
	
	\author{
 Ali~Gholami,~\IEEEmembership{Student Member,~IEEE,}
 Kai~Wan,~\IEEEmembership{Member,~IEEE,} 
 Tayyebeh Jahani-Nezhad,~\IEEEmembership{Member,~IEEE,}  
Hua~Sun,~\IEEEmembership{Member,~IEEE,}
Mingyue~Ji,~\IEEEmembership{Member,~IEEE,}  
and~Giuseppe Caire,~\IEEEmembership{Fellow,~IEEE}
  \thanks{
A short version of this paper   was accepted  by   the 2024 IEEE International Symposium on Information Theory~\cite{ourpaper}. 
}
\thanks{A.~Gholami,  T.~Jahani-Nezhad, and G.~Caire are with the Electrical Engineering and Computer Science Department, Technische Universit\"at Berlin, 10587 Berlin, Germany (e-mail: \{a.gholami, caire, t.jahani.nezhad\}@tu-berlin.de). The work of A.~Gholami,  T.~Jahani-Nezhad, and G.~Caire was partially funded by the European Research Council under the ERC Advanced Grant N. 789190, CARENET.}
\thanks{
K.~Wan is with the School of Electronic Information and Communications,
Huazhong University of Science and Technology, 430074  Wuhan, China,  (e-mail: kai\_wan@hust.edu.cn). The work of K.~Wan  was partially funded by the   National Natural
Science Foundation of China (NSFC-12141107).}
\thanks{
H.~Sun is with the Department of Electrical Engineering, University of North Texas, Denton, TX 76203, USA (email: hua.sun@unt.edu). The work of H.~Sun was supported in part by NSF under Grant CCF-2007108, Grant CCF-2045656, and Grant CCF-2312228.
}
\thanks{
M.~Ji is with the Department of Electrical and Computer Engineering in the University of Florida, Gainesville, FL, 32611, USA (email: mingyueji@ufl.edu).  The work of M.~Ji was partially funded by National Science Foundation (NSF) Award 2312227 and CAREER Award 2145835. }
	}
	
	\maketitle
	
	\begin{abstract}
	In a typical formulation of the private information retrieval (PIR) problem, a single user wishes to retrieve one out of $ K$ files from $N$  servers without revealing the demanded file index to any server.
	This paper formulates an extended model of PIR, referred to as multi-message private computation (MM-PC), where 	instead of retrieving a single file, the user wishes to retrieve $P>1$ linear combinations of files  while preserving the privacy of the demand information. 
The MM-PC problem is a generalization of the private computation (PC) problem (where the user requests one linear combination of the files), and the multi-message private information retrieval (MM-PIR) problem   (where the user requests $P>1$ files). 
A baseline achievable scheme repeats the optimal PC scheme by Sun and Jafar $P$ times, or treats each possible demanded linear combination as an independent file and then uses the near optimal MM-PIR scheme   by Banawan and Ulukus. 
 In this paper, we propose a new MM-PC scheme that significantly improves upon the baseline schemes. In doing so, we design the queries inspired by the structure in the cache-aided scalar linear function retrieval scheme by Wan {\it et al.}, which leverages the dependency between linear functions  to reduce the amount of communications. To ensure the decodability of our scheme, we propose a new method to benefit from the existing dependency, referred to as the sign assignment step. In the end, we use Maximum Distance Separable matrices to code the queries, which allows the reduction of download from the servers, while preserving privacy. By the proposed schemes, we characterize the capacity  within a multiplicative factor of $2$.
 
 
	\end{abstract}
	
	
	\section{Introduction}
	

In the private information retrieval (PIR) problem~\cite{sun2017capacity}, a user wishes to download a file by sending different queries to a group of $N$ non-colluding servers each storing the same $K$ files, while keeping the identity of the desired file secret from the servers. The information-theoretic capacity is defined as the maximum number of bits of desired information decoded per one bit of downloaded information. The authors in~\cite{sun2017capacity} show that the capacity of PIR is given by $\frac{1-1/N}{1-1/N^K}$.

Following the seminal PIR result in~\cite{sun2017capacity}, a large number of works have considered extended models of PIR. In particular,
in ~\cite{sun2018capacity,MirmohseniPFR}, the problem of \textbf{private computation (PC)} is proposed. In general, linear and multivariate polynomial operations are widely used as fundamental primitives for building the complex queries that support online big-data analysis and data mining procedures. In these scenarios,
it is too resource-consuming to locally download all input variables in order to compute the desired output value. Based on this motivation, 
the PC problem is considered in~\cite{sun2018capacity,MirmohseniPFR}, where instead of retrieving a single file, the user requests a (scalar) linear combination of the files among $M$ possible linear combinations, where each linear combination is called a message. 
An optimal PC scheme has been proposed in~\cite{sun2018capacity}. It is interesting to note that the capacity of the PC problem is exactly the same as that of the PIR problem, which is independent of $M$. 
Several extended models of the PC problem have been considered, including PC with coded storages at the servers~\cite{ObeadISIT,ObeadPFRwithcoded,KarpukPFR}, private   sequential function retrieval~\cite{TahmasebisequencePFR} (where the user wants to compute a fixed set of linear combinations while hiding the computation order), PC with polynomial functions~\cite{Obeadpolynomial2022,Ravivlangrange}, cache-aided PC~\cite{Yan2022cachePIR}, single-server PC~\cite{Heidarzadehsingle2022}, and more. 

Another line of work in PIR is the \textbf{multi-message PIR (MM-PIR)} proposed in ~\cite{banawan2018multi}. Instead of retrieving a single file,  in the MM-PIR problem, the user aims to retrieve $P>1$ files. A near-optimal MM-PIR scheme has been proposed in~\cite{banawan2018multi}.   It is also interesting to note that, even if the requested files are independent, designing the MM-PIR scheme  by jointly considering the multi-request (as in~\cite{banawan2018multi}) leads to 
 a significant increase in the retrieval rate compared to simply repeating the Sun and Jafar PIR scheme $P$ times. Other works related to MM-PIR include~\cite{shariatpanahi2018multi}, where the problem assuming that the user has private side information is studied, and ~\cite{li2018single}, ~\cite{heidarzadeh2018capacity}, which consider the MM-PIR problem with side information in the single-server case.


In this paper, we formulate a new problem, referred to as the MM-PC problem, which covers the PC and MM-PIR problems as special cases. In this setting, there are $N$ non-colluding servers, each storing a library of $M$ messages with arbitrary linear dependencies, of which $K$ are linearly independent. The user wants to retrieve a set of $P$ linearly independent messages from the servers, while keeping the identity of the requested messages secret from each server.

Two recent problems are  similar to our formulated MM-PC problem.  The private linear transformation (PLT) problem has been considered in~\cite{EsmatiISIT2021jointprivacy,Esmati2021ISITindividual,Kazemi2021XVII,Esmati2021XVII}. In the PLT problem, the user also wants to retrieve $P$ linear combinations of $\tilde{K}<K$ files while preserving the privacy of the indices of the $\tilde{K}$ files. The private distributed computing problem has been considered in~\cite{Chang2019ITWprivate,Kim2020TIFSprivate,akbari2021secure,Aliasgari2020TIFSprivate,Li2022TIFSprivate,Yu2021TCOMcoded,ZhuJSAIT2022systematic,YangISIT2021PRIVATE}, where the user wants to compute a matrix multiplication $A B_i$ where $B_1,B_2,\ldots$ are  matrices with uniform i.i.d. elements while preserving the privacy of the index $i$. 
In our considered problem, the results for the above two problems cannot be applied (or are highly inefficient).\footnote{\label{foot:inefficiency}More precisely, the PLT schemes cannot be applied to our problem, since in our problem the linear combinations are over all files, and we aim to preserve the privacy of the coefficient matrix instead of chosen files. The private distributed computing schemes are very inefficient to be applied to our problem, since we should treat each possible set of linear combinations as an ``independent'' demand matrix, and thus there is a huge number of such possible demand matrices.}  A very recent work on private multiple linear computation appeared in~\cite{zhu2024privatelc}, where the problem is to compute multiple linear combinations of some messages, which are replicated on multiple servers, by considering the case of colluding and non-responsive servers. While keeping the privacy of the requests, the scheme  in~\cite{zhu2024privatelc}   attains a tradeoff between the communication  and  computation costs, where the communication cost also includes the upload cost. 
However, when applied to the MM-PC problem considered here, the scheme of~\cite{zhu2024privatelc} achieves the rate $\frac{N-1}{N}$, which can also be achieved by repeating the PIR scheme in \cite{shah2014one} $P$ times. The main challenge addressed in our work is how to improve the repetition strategy.

\paragraph*{Contributions}

An achievable scheme by a direct extending of the optimal PC scheme in \cite{sun2018capacity} or the near optimal MM-PIR scheme in~\cite{banawan2018multi}, is proposed  which we refer to as the baseline scheme.

However, the direct combination of the PC scheme in \cite{sun2018capacity} and the   MM-PIR scheme in~\cite{banawan2018multi} is not possible. Hence, we propose a new scheme that improves over the baseline scheme, by leveraging some features of the optimal PC and near optimal MM-PIR schemes and incorporating some non-trivial novel ideas. More precisely, 
while each message is divided into multiple symbols and the queries are essentially linear combinations of these symbols, to exploit the dependency between messages, we may need to assign a specific sign to each symbol involved, referred to as sign assignment. To ensure decodability,  and inspired by \cite{wan2021optimal}, we propose a new sign assignment method which makes some of the queries linear combinations of others, and then by using Maximum Distance Separable (MDS) coding, we can reduce the amount of download, while preserving symmetry and thus privacy.  It is essential to mention that the redundancy appears as a result of the novel sign and index assignment method. Numerical evaluations show that the   improved scheme provides large performance gains   with respect to the baseline scheme for a wide range of system parameters. 
 
 \noindent{\bf Notation:}
 For $a\in \mathbb{N}$ the notation $[a]$ represents set $\{1,\dots,a\}$, and notation 
$[a:b]$ for $a,b\in\mathbb{N}$ represents set $\left\{ a,a+1,\ldots,b\right\}$. In addition, we denote the difference of two sets $\mathcal{A}$, $\mathcal{B}$ as $\mathcal{A}\setminus\mathcal{B}$, that means the set of elements which belong to $\mathcal{A}$ but not $\mathcal{B}$.

	\section{Problem setting}
	Consider $N$ non-colluding servers with $K$ files which are replicated on all servers. 	
For each $i\in[K]$, the $i^{\text{th}}$ file is a vector of large enough size  $L$, denoted by $W_{d_i}\in \mathbb{F}_q^{L}$, whose symbols take on values over a finite field $\mathbb{F}_q$. Additionally, files are independently and randomly generated with i.i.d. symbols such that 
	\begin{subequations} 
	\begin{align}
		&H(W_{d_1}) = \cdots = H(W_{d_K}) = L, \\
		&H(W_{d_1}, \cdots, W_{d_K}) = H(W_{d_1}) + \cdots + H(W_{d_K}).
	\end{align}
	\end{subequations} 
	Note that in this paper, the $\log$ used for information measures in the entropy function is base-$q$. 
 A user wants to retrieve  $P$ of $M$ possible messages from the servers, where each message is a linear combination of the $K$	files.  
For each $m\in [M]$, the $m^{\text{th}}$ message is defined as,
\begin{subequations} 
	\begin{align}
		W_m &:= \mathbf{v}_m[W_{d_1}, \dots, W_{d_K}]^T \\
		& = v_m(1) W_{d_1} + \dots + v_m(K) W_{d_K}, 
	\end{align}
	\end{subequations} 
	where $v_m(i)$ is the $i^{\text{th}}$ entry of the coefficient vector $\mathbf{v}_m$ for $i\in[K]$, and all operations are taken in $\mathbb{F}_q$. 
Without loss of generality, we assume that	$M \geq K$ and  the first $K$ messages are replicas of the $K$ independent files, i.e., $(W_1, \dots, W_K) = (W_{d_1}, \dots, W_{d_K})$. For the sake of future convenience, each message in $W_1, \dots, W_K$ is called an independent message; each other message is called a dependent message, since it is a linear combination of independent messages.
	
	Unlike~\cite{sun2018capacity} where the user requires only one message, in the MM-PC problem, the user 
	 privately generates a set of  $P$ indices $\mathcal{I} = \{\theta_1, \dots, \theta_P\}$, where $\mathcal{I}\subset [M]$  and $\theta_i \neq \theta_j$ for each $i, j \in [P]$ where $i\neq j$.  The user wishes to compute $W_{\mathcal{I}} := (W_{\theta_1}, \dots, W_{\theta_P})$ while keeping $\mathcal{I}$ secret from each server. 
Without loss of generality, we assume that 	 $W_{\theta_1}, \dots, W_{\theta_P}$ are linearly independent; otherwise, we can just reduce $P$ and let the user demand linearly independent combinations. 
	  To do so, the user generates $N$ queries $Q_1^{\mathcal{I}}, \dots, Q_N^{\mathcal{I}}$ and sends each $Q_n^{\mathcal{I}}$ to the corresponding server.   These queries are generated when the user has no knowledge of the realizations of the messages, so the queries should be independent of the messages, i.e.,
	\begin{align}
		I(Q_1^{\mathcal{I}}, \dots, Q_N^{\mathcal{I}} ; W_1, \dots, W_M) = 0. 
	\end{align}
	
	Upon receiving $Q_n^{\mathcal{I}}$, each server $n \in [N]$ generates and sends the answer $A_n^{\mathcal{I}}$ which is a function of $Q_n^{\mathcal{I}}$ and $W_1, \dots, W_M$, i.e.,  
	\begin{align}
		H(A_n^{\mathcal{I}} | Q_n^{\mathcal{I}}, W_1, \dots, W_M) = 0, n \in [N]. 
	\end{align}
Finally, the user must retrieve the desired $W_{\mathcal{I}}$ from the servers’ answers $A_n^{\mathcal{I}}$ and the queries $ Q_n^{\mathcal{I}}$ with vanishing error\footnote{The MM-PC scheme proposed in this paper however, has zero probability of error.}, i.e., 
	\begin{align}
H (W_{\mathcal{I}}|A_1^{\mathcal{I}}, \dots, A_N^{\mathcal{I}}, Q_1^{\mathcal{I}}, \dots, Q_N^{\mathcal{I}})=o(L),\label{eq:decodability}  
\end{align}	 
where $\lim_{L\to \infty } o(L)/L=0$. 
	
	The MM-PC scheme should be designed to keep the demand information $\mathcal{I}$ secret from all servers; i.e., the following privacy 
	 constraint must be satisfied, 
	\begin{align} \label{privacy}
		(Q_n^{\mathcal{I}_1}, A_n^{\mathcal{I}_1}, W_1, \dots, W_M) \sim (Q_n^{\mathcal{I}_2}, A_n^{\mathcal{I}_2}, W_1, \dots, W_M),
	\end{align}
	for all $\mathcal{I}_1, \mathcal{I}_2 \in \Omega$ and all servers $n \in [N]$, where $\Omega$ is the set of all possible $\mathcal{I}$, and $\sim$ indicates that these two random vectors follow the same distribution.
	
	The MM-PC \textit{rate}, denoted by $R$, is defined as the number of symbols recovered collectively from all the demanded messages per one downloaded symbol,
	\begin{align} \label{eq:rate}
		R :=\frac{PL}{D},
	\end{align}
	where $D$ is the expected value over random queries of the total downloaded symbols from all the servers by the user. The objective is to find to    the supremum of all achievable rates, denoted by $R^{\star}$. 
	\section{Main Results}
	In this section, we present the baseline scheme and the main results for the proposed MM-PC problem.

	\begin{theorem}[Baseline scheme] \label{cor:repetition}
		For the MM-PC problem, the following rate is achievable, 
		\begin{align}
			R_{1} = \max \left\{ \frac{1-\frac{1}{N}}{1-(\frac{1}{N})^K} + \frac{(P-1)(N-1)}{N^M \left(1-(\frac{1}{N})^K\right)}, C_{M,P} \right\}, \label{eq:rep scheme rate}
		\end{align}
		where $C_{M,P}$ represents the achieved rate of the MM-PIR scheme in~\cite{banawan2018multi} with $M$ files in the library and $P$ requests from the user. 
	\end{theorem}

    The first rate in~\eqref{eq:rep scheme rate} is achieved by simply repeating the single-message private computation scheme (PC scheme) in \cite{sun2018capacity} by removing some redundant   symbols downloaded in the first round of PC. The second rate in~\eqref{eq:rep scheme rate} is achieved  by 
    treating each possible demanded linear combination as an independent message, and then using the MM-PIR scheme in~\cite{banawan2018multi}. Next, we show that the order optimality of the baseline scheme.   
    \begin{theorem} [order-optimality of the baseline scheme] \label{thm:order-optimal}
        The baseline scheme in Theorem \ref{cor:repetition} is order-optimal within a multiplicative gap of $2$.
    \end{theorem}
    The detailed description of the baseline scheme for Theorem~\ref{cor:repetition} and its order optimality proof could be found in  Appendix \ref{sec:rep}.

    \begin{remark} [asymptotic optimality of the baseline scheme]
        Based on Theorem \ref{thm:order-optimal}, the gap between the optimal scheme and the baseline scheme is bounded by $\frac{1}{1-\frac{1}{N}}$, which for large $N$ converges to $1$. This shows the asymptotic optimality of the baseline scheme for large $N$. 
    \end{remark}
    
Even though the baseline scheme (which treats each linear function as one file and uses the optimal PIR scheme $P$ times) is order-optimal within $2$, it can be further improved by carefully leveraging the connection among the  linear functions.
    The achieved rate of our improved scheme is listed in the following theorem, and its description is presented in Section~\ref{sec:achievable}.
	\begin{theorem}[Proposed scheme] 
	\label{thm:improved scheme}
		For the MM-PC problem, in case $P < K$, the following rate is achievable,
		\begin{align}
			R_2 = \frac{P \sum_{i=1}^{M-P+1} \alpha_i \binom{M-P}{i-1}}{\sum_{i=1}^{M-P+1} \alpha_i \left(\binom{M-P}{i} - \binom{M-K}{i} + P \binom{M-P}{i-1}\right)},
		\end{align}
		where $\alpha_{M-P+2} = \cdots = \alpha_M = 0$, $\alpha_{M-P+1} = (N-1)^{M-P}$,  and
		\begin{align} \label{eq:stages}
			\alpha_i = \frac{1}{N-1} \sum_{m=1}^{P} \binom{P}{m} \alpha_{i+m}, \hspace{3mm} i \in [1:M-P].
		\end{align}
	\end{theorem}
 
    The proofs of decodability and privacy of the proposed scheme for Theorem~\ref{thm:improved scheme} are provided in Appendices~\ref{sec:decodability} and~\ref{sec:privacy}, respectively. 
Fig.~\ref{fig:compare} compares the baseline scheme with the proposed scheme, for the case where $K=7$, $N=2$, $M\in\{10,15\}$, and $P\in [2:6]$. 
As shown in Fig.~\ref{fig:compare}, when $P=2$, the baseline scheme is   slightly better than the proposed scheme; when $P>2$, the improvement over the baseline scheme becomes more significant as $P$ increases. Note that the rate of the proposed scheme has very little dependence on $M$, since the solid red line and the dashed red line almost coincide. This can be also observed  in Fig.~\ref{fig:compare_msweep}, sweeping on $M$ does not change much  the rate for the proposed scheme. 

	\begin{figure} 
		\centering 
		\includegraphics[width=90mm]{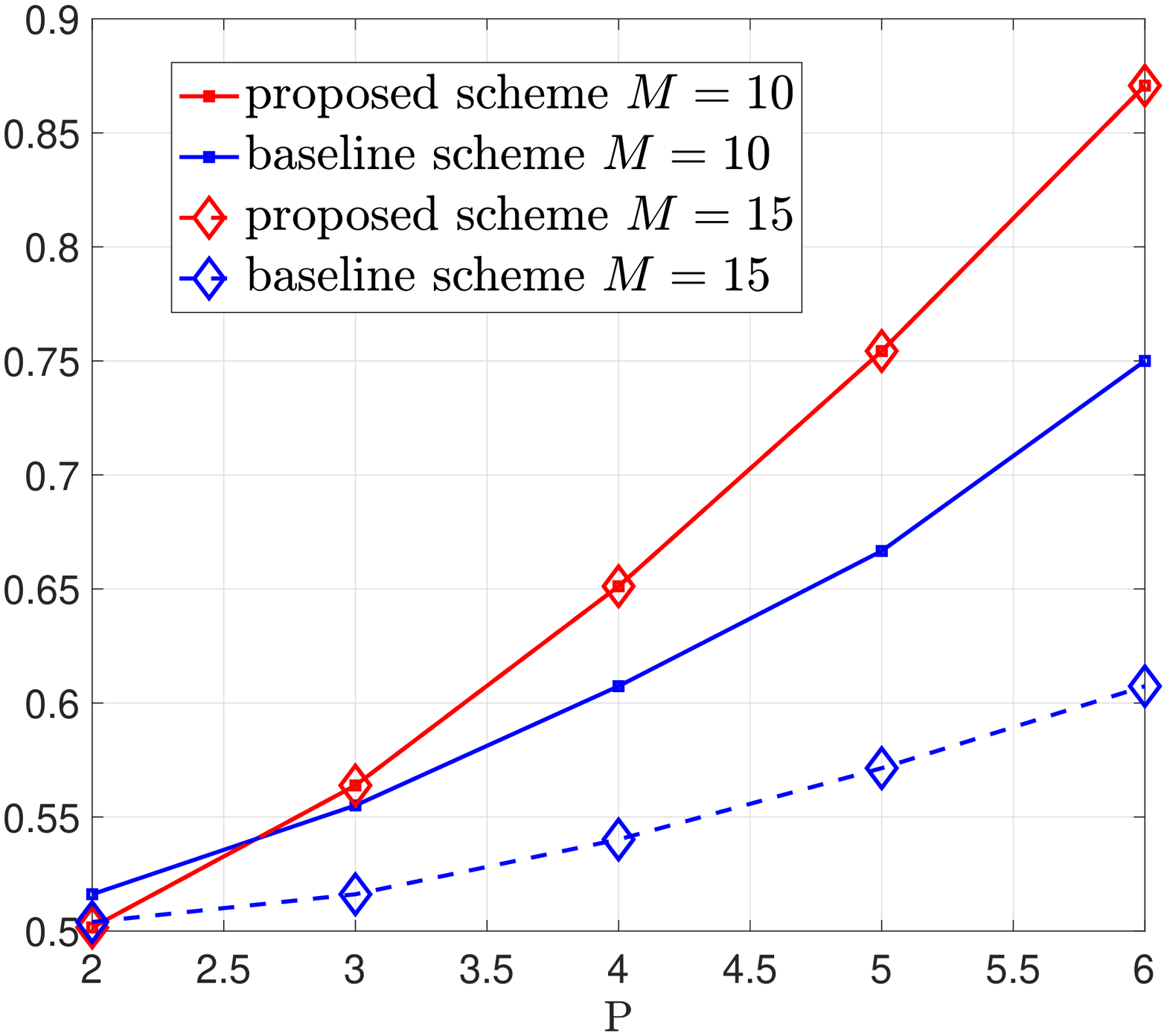}
		\caption{Comparison of the rates. The red lines are for our proposed scheme and the blue ones for the baseline scheme. The parameters are $K=7, N=2$. $P$ and $M$ changes as in the figure.}
  \label{fig:compare}
		\vspace{-5mm}
	\end{figure}

        \begin{figure} 
		\centering 
		\includegraphics[width=90mm]{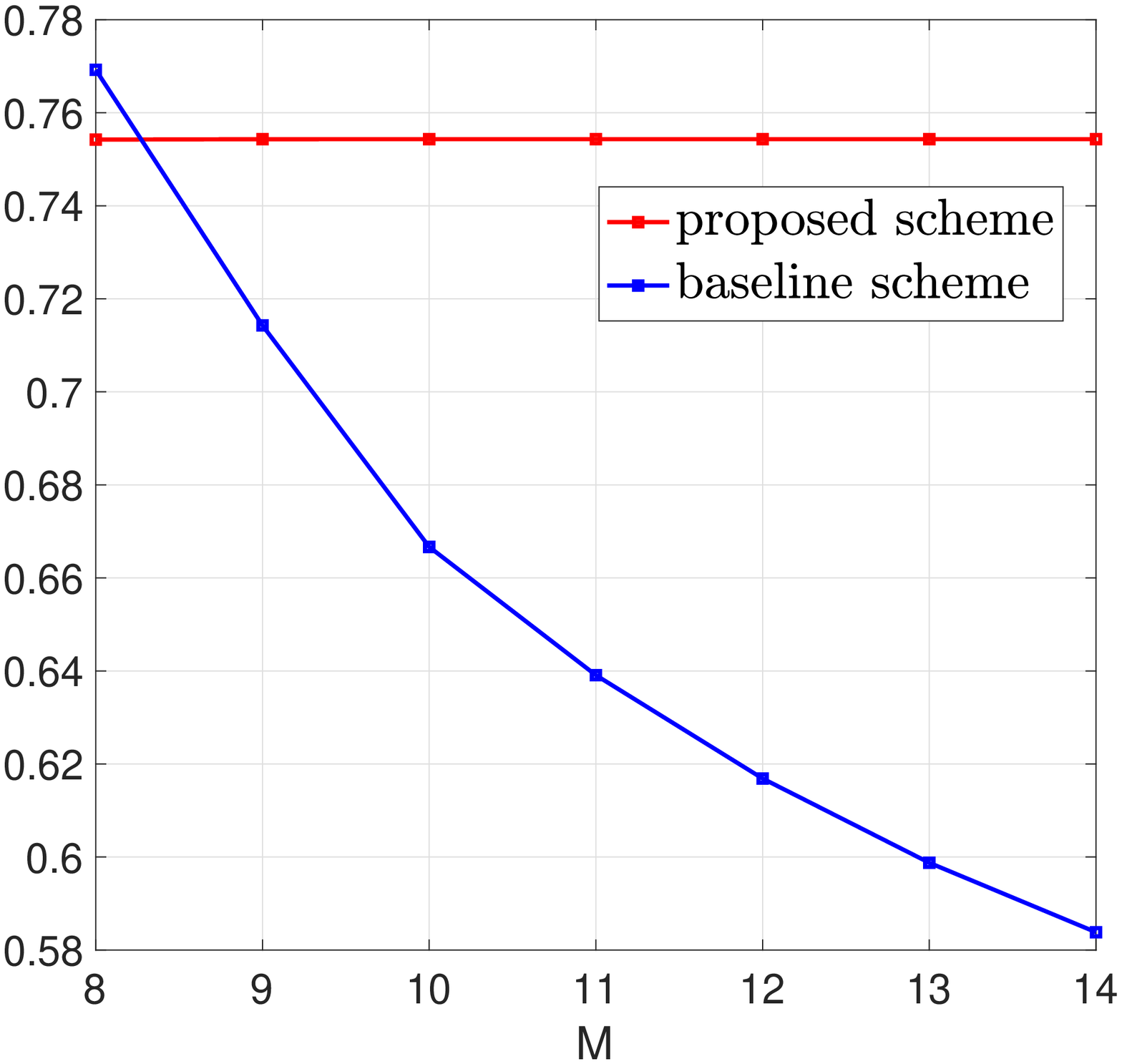}
		\caption{Comparison of the rates. The red lines are for our proposed scheme and the blue ones for the baseline scheme. The parameters are $K=7, N=2$, $P=5$ and $M$ changes as in the figure. As evident, the dependency of the proposed scheme on $M$ is almost zero.}
  \label{fig:compare_msweep}
		\vspace{-5mm}
	\end{figure}

        \section{New Proposed MM-PC Scheme Through an Example} \label{sec:exp achievable}
        For the sake of clarity, we illustrate the proposed scheme through an example and provide the general description in the next section.
        

        The step-by-step example in this section considers $M=5$, $K=3$, $P=2$, and $N=2$. 
        The messages are denoted by letters $\{a,b,c,d,e\}$, where $\{a,b,c\}$ are the independent files and $\{d,e\}$ are any desired linear combinations of the independent files in the given finite field $\mathbb{F}_q$. In this example, the demanded files are $\mathcal{I}=\{a,b\}$. Each message is partitioned into $L=68$ symbols and the $i^{\text{th}}$ symbol of each message is denoted by subscript (index) $i$, e.g., $a_i$ denotes the $i^{\text{th}}$ symbol of message $a$.  Note that the general relation between the system parameters and the subpacketization level $L$ follows $L = N \sum_{i=1}^{M-P+1} \alpha_i \binom{M-P}{i-1}$,  where the exact calculation appears in Appendix~\ref{sec:decodability}. 

\textbf{Step 1: Permutation and Relabeling. } A permutation function $\pi(\cdot)$ on $[L]$ is chosen uniformly at random over all the $L!$ possibilities. The symbols of every message are permuted by $\pi$. For simplicity, the permuted messages are denoted by the same letters $\{a,b,c,d,e\}$, e.g., message $a=(a_1,a_2, ..., a_{68})$ turns into $a=(a_{\pi(1)}, a_{\pi(2)}, ..., a_{\pi(68)})$. Furthermore, we define the variables $\sigma_i \in \{-1,1\}, \forall i\in [L]$, referred to as multiplicative factors, each chosen uniformly i.i.d. For all messages, the symbol of position $i$ is multiplied by $\sigma_i$. For example, the message $a$ is transformed into $a=(\sigma_1 a_{\pi(1)}, \sigma_2 a_{\pi(2)}, ..., \sigma_{68} a_{\pi(68)})$. For the ease of description, in this example we assume that  the permutation   is  $(1,2,\ldots,68)$ and 
    that    $\sigma_i=1, \forall i \in [68]$. 
    
     We also perform a relabeling on the message labels $1,...,M$, such that the first $P$ labels of messages (i.e., $W_1,...,W_P$) are the demanded messages. Furthermore, we change the set of independent messages, such that they contain the $P$ demanded messages. It will be proved in Appendices \ref{sec:decodability} and \ref{sec:privacy} that these actions will not hurt the decodability and privacy of our scheme. In this example, this is already the case and there is no change needed.

	\textbf{Step 2: Query Structure and Number of Repetitions.} The queries to servers are linear combinations of symbols from different messages. They are categorized into multiple \textit{rounds}, where round $i$ contains queries summing $i$ different symbols. Each round itself is also split into multiple \textit{stages}. Each stage of round $i$ contains all $\binom{M}{i}$ choices of $i$ messages from the total $M$. For instance for a stage of round $2$, the queries are of the form $\{a_*+b_*,a_*+c_*,a_*+d_*,a_*+e_*,b_*+c_*,b_*+d_*,b_*+e_*,c_*+d_*,c_*+e_*,d_*+e_*\}$, which covers all $\binom{5}{2}=10$ ways of choosing $2$ messages from the total $5$. Note that the subscript $*$  denotes some specific symbol index. The number of stages of round $i$ denoted by $\alpha_i$, follows \eqref{eq:stages}. The explanation to calculate $\alpha_i$ is provided in Section \ref{sec:achievable}, Step 2. For our example we have $\alpha_5=0, \alpha_4=1, \alpha_3=2, \alpha_2=5, \alpha_1=12$.\footnote{The number of stages calculated here is completely different from that of \cite{banawan2018multi}. The main reason is that in the scheme~\cite{banawan2018multi}, every query containing symbols of demanded messages contributes to decoding new demanded symbols, while in the proposed scheme because of the special index assignment, designed in cooperation with the sign assignment, this is not possible. }

	\textbf{Step 3: Initialization. } This step corresponds to queries of round $1$ (single symbols). Since $\alpha_1=12$, from each server the user queries $12$ symbols of each message, depicted in Table \ref{tab:round1}. 
	
	
	\begin{table}
        \vspace{+3mm}
		\centering
		\begin{tabular}{| c | c | c | c |} 
			\hline
			Round & Stage & Server $1$ & Server $2$ \\
			\hline
            \multirow{3}{*}{round $1$}
			& stage $1$ &
			$a_1,b_1,c_1,d_1,e_1$ & $a_{13},b_{13},c_{13},d_{13},e_{13}$ \\
			\cline{2-4}
			& \vdots & \vdots & \vdots \\
			\cline{2-4}
			& stage $12$ & $a_{12},b_{12},c_{12},d_{12},e_{12}$ & $a_{24},b_{24},c_{24},d_{24},e_{24}$ \\
			\hline
		\end{tabular}
		\caption{Round $1$ queries.}
  \label{tab:round1}
	\end{table}
	
	
	\textbf{Step 4: Index Assignment. }
    Since the general structure of queries is known from Step~2, i.e., the number of stages in each round and that each stage of round $i$ contains all possible $i$-sums, we need to determine the symbol indices for each query. The index assignment is inspired by the delivery phase of the coded caching scheme in~\cite{maddah2014fundamental}. Note that the construction of this coded caching scheme is completely symmetric over files if the number of files is equal to the number of users and each user requests a distinct file.
    
    Consider the first stage of round $2$ queries to server $1$.
    There are $\binom{5}{2}=10$ queries of the form $\{a_*+b_*,a_*+c_*,a_*+d_*,a_*+e_*,b_*+c_*,b_*+d_*,b_*+e_*,c_*+d_*,c_*+e_*,d_*+e_*\}$.
    
    Let us first determine the indices of the   queries in $\{a_*+c_*,a_*+d_*,a_*+e_*\}$,
    where the indices of $c_*,d_*, e_*$  in $\{a_*+c_*,a_*+d_*,a_*+e_*\}$ should be the same (treated as the side information to decode new symbols of message $a$). In addition, these symbols $c_*,d_*, e_*$  should have been downloaded previously. Hence, we can let these queries be 
     $\{a_{25}+c_{13},a_{26}+d_{13},a_{27}+e_{13}\}$, where $\{25,26,27\}$ are new indices of message $a$ (the first $24$ indices are already used in round $1$) and the remaining parts are symbols with index $13$, already received from queries to server $2$ in round $1$. 
      Similarly in the same stage, for the demanded message $b$ the queries should be $\{b_{25}+c_{14},b_{26}+d_{14},b_{27}+e_{14}\}$, which use symbols  $c_{14}, d_{14}, e_{14}$ as side information. 

    We then determine the indices in the remaining queries $\{a_*+b_*,c_*+d_*,c_*+e_*,d_*+e_*\}$. 
    Recall that the symbols which are treated as side information to decode $a$ are with index $13$, and the symbols which are treated as side information to decode $b$ are with index $14$. Hence 
     $a_*+b_*$ should be  $a_{14}+b_{13}$. In addition, the symbols which are treated as side information to (virtually) decode $c,d,e$ are with indices $25,26,27$, respectively. Thus these remaining  queries should be  
  $\{a_{14}+b_{13},c_{26}+d_{25},c_{27}+e_{25},d_{27}+e_{26}\}$. 
  
  For the first stage of round $2$ queries to server $2$, the same process repeats using symbols with indices $1$ and $2$ acting as side information and decoding new symbols with indices $\{28,29,30\}$ for messages $a$ and $b$. The other $4$ stages of round $2$ follow the same procedure. 
  The  queries of round $2$ are given in Table \ref{tab:round2example}. 
    
    \begin{table}
    \vspace{+3mm}
		\centering
		\begin{tabular}{| c | c | c | c |}
			\hline
			Round & Stage & Server $1$ & Server $2$ \\
			\hline
			\multirow{50}{*}{round $2$}
			& \multirow{10}{*}{stage $1$}
                &
                        $a_{25}+c_{13}$ & $a_{28}+c_{1}$ \\
 			        &&$a_{26}+d_{13}$ & $a_{29}+d_{1}$ \\
 			        &&$a_{27}+e_{13}$ & $a_{30}+e_{1}$ \\
                        &&$b_{25}+c_{14}$ & $b_{28}+c_{2}$ \\
 			        &&$b_{26}+d_{14}$ & $a_{29}+d_{2}$ \\
                        &&$b_{27}+e_{14}$ & $b_{30}+e_{2}$ \\
                        &&$a_{14}+b_{13}$ & $a_{2}+b_{1}$ \\
                        &&$c_{26}+d_{25}$ & $c_{29}+d_{28}$ \\
                        &&$c_{27}+e_{25}$ & $c_{30}+e_{28}$ \\
                        &&$d_{27}+e_{26}$ & $d_{30}+e_{29}$ \\
			\cline{2-4}
			& \multirow{10}{*}{stage $2$}
                &
 					$a_{31}+c_{15}$ & $a_{34}+c_{3}$ \\
 					&&$a_{32}+d_{15}$ & $a_{35}+d_{3}$ \\
 					&&$a_{33}+e_{15}$ & $a_{36}+e_{3}$ \\
                        &&$b_{31}+c_{16}$ & $b_{34}+c_{4}$ \\
 					&&$b_{32}+d_{16}$ & $a_{35}+d_{4}$ \\
                        &&$b_{33}+e_{16}$ & $b_{36}+e_{4}$ \\
                        &&$a_{16}+b_{15}$ & $a_{4}+b_{3}$ \\
                        &&$c_{32}+d_{31}$ & $c_{35}+d_{34}$ \\
                        &&$c_{33}+e_{31}$ & $c_{36}+e_{34}$ \\
                        &&$d_{33}+e_{32}$ &$d_{36}+e_{35}$ \\
			\cline{2-4}
			& \multirow{10}{*}{stage $3$}
                &
 					$a_{37}+c_{17}$ & $a_{40}+c_{5}$ \\
 					&&$a_{38}+d_{17}$ & $a_{41}+d_{5}$ \\
 					&&$a_{39}+e_{17}$ & $a_{42}+e_{5}$ \\
                        &&$b_{37}+c_{18}$ & $b_{40}+c_{6}$ \\
 					&&$b_{38}+d_{18}$ & $a_{41}+d_{6}$ \\
                        &&$b_{39}+e_{18}$ & $b_{42}+e_{6}$ \\
                        &&$a_{18}+b_{17}$ & $a_{6}+b_{5}$ \\
                        &&$c_{38}+d_{37}$ & $c_{41}+d_{40}$ \\
                        &&$c_{39}+e_{37}$ & $c_{42}+e_{40}$ \\
                        &&$d_{39}+e_{38}$ & $d_{42}+e_{41}$ \\
			\cline{2-4}
			& \multirow{10}{*}{stage $4$}
                &
 					$a_{43}+c_{19}$ & $a_{46}+c_{7}$ \\
 					&&$a_{44}+d_{19}$ & $a_{47}+d_{7}$ \\
 					&&$a_{45}+e_{19}$ & $a_{48}+e_{7}$ \\
                        &&$b_{43}+c_{20}$ & $b_{46}+c_{8}$ \\
 					&&$b_{44}+d_{20}$ & $a_{47}+d_{8}$ \\
                        &&$a_{20}+b_{19}$ & $b_{48}+e_{8}$ \\
                        &&$a_{6}+b_{5}$ & $a_{8}+b_{7}$ \\
                        &&$c_{44}+d_{43}$ & $c_{47}+d_{46}$ \\
                        &&$c_{45}+e_{43}$ & $c_{48}+e_{46}$ \\
                        &&$d_{45}+e_{44}$ & $d_{48}+e_{47}$ \\
			\cline{2-4}
			& \multirow{10}{*}{stage $5$}
                &
 					$a_{49}+c_{21}$ & $a_{52}+c_{9}$ \\
 					&&$a_{50}+d_{21}$ & $a_{53}+d_{9}$ \\
 					&&$a_{51}+e_{21}$ & $a_{54}+e_{9}$ \\
                        &&$b_{49}+c_{22}$ & $b_{52}+c_{10}$ \\
 					&&$b_{50}+d_{22}$ & $a_{53}+d_{10}$ \\
                        &&$b_{51}+e_{22}$ & $b_{54}+e_{10}$ \\
                        &&$a_{22}+b_{21}$ & $a_{10}+b_{9}$ \\
                        &&$c_{50}+d_{49}$ & $c_{53}+d_{52}$ \\
                        &&$c_{51}+e_{49}$ & $c_{54}+e_{52}$ \\
                        &&$d_{51}+e_{50}$ & $d_{54}+e_{53}$ \\
			\hline
		\end{tabular}
		\caption{Round 2 of queries.}
  \label{tab:round2example}
	\end{table}

    A stage of round $3$ contains all $\binom{5}{3}$ ways of choosing $3$ messages out of $5$, i.e., $\{a_*+b_*+c_*,a_*+b_*+d_*,a_*+b_*+e_*,a_*+c_*+d_*,a_*+c_*+e_*,a_*+d_*+e_*,b_*+c_*+d_*,b_*+c_*+e_*,b_*+d_*+e_*,c_*+d_*+e_*\}$. 
    
    The queries $\{a_*+c_*+d_*,a_*+c_*+e_*,a_*+d_*+e_*\}$, 
    are used to decode new symbols of $a$. 
    Consider the first stage of round $3$ queries to server $1$. The new indices for the symbols of the demanded message $a$ are $\{55,56,57\}$, since the first $54$ symbols of $a$ have already appeared in the first two rounds. The side information part is duplicated from the first stage of round $2$ queries to server $2$, i.e., $\{c_{29}+d_{28}, c_{30}+e_{28}, d_{30}+e_{29}\}$. Therefore, these queries would be $\{a_{55}+c_{29}+d_{28}, a_{56}+c_{30}+e_{28}, a_{57}+d_{30}+e_{29}\}$. Similarly for decoding new symbols of $b$, the queries are $\{b_{55}+c_{35}+d_{34}, b_{56}+c_{36}+e_{34}, b_{57}+d_{36}+e_{35}\}$, where the side information parts are duplicated from the second stage of round $2$ queries to server $2$.  
    One observes that when $c$ and $d$ appear together in a query, the indices of the other involved symbols  are the same, i.e., in $\{a_{55}+c_{29}+d_{28},b_{55}+c_{35}+d_{34}\}$, both $a$ and $b$ have index $55$. When $a$ and $d$ appear together, the indices of the other involved symbols  are both $29$; see $\{a_{55}+c_{29}+d_{28},a_{57}+d_{30}+e_{29}\}$.
    One can verify that the same structure holds for any  two messages appearing together. It is also interesting to see that this phenomenon also exists in the coded caching scheme in~\cite{maddah2014fundamental}, which is the core to preserve the privacy.
    
    Among the remaining queries $\{a_*+b_*+c_*,a_*+b_*+d_*,a_*+b_*+e_*,c_*+d_*+e_*\}$, first consider $a_*+b_*+c_*$. To determine the index of $a$, we search for a query containing both $b$ and $c$, e.g., $b_{55}+c_{35}+d_{34}$. 
    Since $b_{55}+c_{35}+d_{34}$ and $a_*+b_*+c_*$ both contain $b$ and $c$, the indices of $a$ and $d$ should be the same, which is $34$.
To determine the index of $b$, since in the query $a_{55}+c_{29}+d_{28}$  containing both $a$ and $c$, $b$ should  have the same index as $d_{28}$, which is 
$28$. To determine the index of $c$,  there has not been any query containing both $a$ and $b$ in this stage yet. Note that $a_{34}$ and $b_{28}$ have already been recovered by the user from the previous transmissions. Hence, we use symbol $c_{23}$ transmitted  in the first round of queries to server $2$, such that the whole sum
$a_{34}+b_{28}+c_{23}$ could be recovered by the previous rounds and thus redundant, which could be removed later in Step 6 to reduce the transmissions.  
  Let us then consider the indices of $a_*+b_*+d_*$. $a_*$ should have the same index as $c_{35}$ in $b_{55}+c_{35}+d_{34}$. $b_*$ should have the same index as $c_{29}$ in $a_{55}+c_{29}+d_{28}$. $d_*$ should have the same index as $c_{23}$ in $a_{34}+b_{28}+c_{23}$. Thus we fix  $a_*+b_*+d_*=a_{35}+b_{29}+d_{23}$. 
    Similarly,    we can subsequently determine the indices of $a_*+b_*+e_*,c_*+d_*+e_*$ as $a_{36}+b_{30}+e_{23},c_{57}+d_{56}+e_{55}$. 
    
   By this approach,  Round $3$ and $4$ queries are obtained as shown in Table \ref{tab:rounds34example}.

    \begin{table}
    \vspace{+3mm}
		\centering
            \resizebox{\columnwidth}{!}{%
		\begin{tabular}{| c | c | c | c |}
			\hline
                Round & Stage & Server $1$ & Server $2$ \\
                \hline
			\multirow{20}{*}{round $3$}
			& \multirow{10}{*}{stage $1$}
			& $a_{55}+c_{29}+d_{28}$ & $a_{58}+c_{26}+d_{25}$ \\
			& & $a_{56}+c_{30}+e_{28}$ & $a_{59}+c_{27}+e_{25}$ \\
			& & $a_{57}+d_{30}+e_{29}$ & $a_{60}+d_{27}+e_{26}$ \\
			& & $b_{55}+c_{35}+d_{34}$ & $b_{58}+c_{32}+d_{31}$ \\
			& & $b_{56}+c_{36}+e_{34}$ & $b_{59}+c_{33}+e_{31}$ \\
			& & $b_{57}+d_{36}+e_{35}$ & $b_{60}+d_{33}+e_{32}$ \\
			& & $a_{34}+b_{28}+c_{23}$ & $a_{31}+b_{25}+c_{11}$ \\
			& & $a_{35}+b_{29}+d_{23}$ & $a_{32}+b_{26}+d_{11}$ \\
			& & $a_{36}+b_{30}+e_{23}$ & $a_{33}+b_{27}+e_{11}$ \\
			& & $c_{57}+d_{56}+e_{55}$ & $c_{60}+d_{59}+e_{58}$ \\
			\cline{2-4}
			& \multirow{10}{*}{stage $2$}
			& $a_{61}+c_{41}+d_{40}$ & $a_{64}+c_{38}+d_{37}$ \\
			& & $a_{62}+c_{42}+e_{40}$ & $a_{65}+c_{39}+e_{37}$ \\
			& & $a_{63}+d_{42}+e_{41}$ & $a_{66}+d_{39}+e_{38}$ \\
			& & $b_{61}+c_{47}+d_{46}$ & $b_{64}+c_{44}+d_{43}$ \\
			& & $b_{62}+c_{48}+e_{46}$ & $b_{65}+c_{45}+e_{43}$ \\
			& & $b_{63}+d_{48}+e_{47}$ & $b_{66}+d_{45}+e_{44}$ \\
			& & $a_{46}+b_{40}+c_{24}$ & $a_{43}+b_{37}+c_{12}$ \\
			& & $a_{47}+b_{41}+d_{24}$ & $a_{44}+b_{38}+d_{12}$ \\
			& & $a_{48}+b_{42}+e_{24}$ & $a_{45}+b_{39}+e_{12}$ \\
			& & $c_{63}+d_{62}+e_{61}$ & $c_{66}+d_{65}+e_{64}$ \\
			\hline
			\multirow{5}{*}{round $4$}
			& \multirow{5}{*}{stage $1$}
			& $a_{67}+c_{60}+d_{59}+e_{58}$ & $a_{68}+c_{57}+d_{56}+e_{55}$ \\
			& & $b_{67}+c_{66}+d_{65}+e_{64}$ & $b_{68}+c_{63}+d_{62}+e_{61}$ \\
			& & $a_{64}+b_{58}+c_{53}+d_{52}$ & $a_{61}+b_{55}+c_{50}+d_{49}$ \\
			& & $a_{65}+b_{59}+c_{54}+e_{52}$ & $a_{62}+b_{56}+c_{51}+e_{49}$ \\
			& & $a_{66}+b_{60}+d_{54}+e_{53}$ & $a_{63}+b_{57}+d_{51}+e_{50}$ \\			
			\hline			
		\end{tabular} }
	\caption{Rounds 3 and 4 of queries.}
 \label{tab:rounds34example} 
\vspace{-5mm}
\end{table}

 	So far, the only operation used in queries is addition. To exploit the dependency between messages, we may need to also use negation, referred to as sign assignment.  It will be proved later that by the proposed sign assignment, 
  in a stage of round $i$, out of $\binom{M}{i}$ total queries, $\binom{M-K}{i}$ of them are redundant and can be written as linear combinations of others. We should point out that this redundancy result is also achieved by the PC scheme in \cite{sun2018capacity}. However, due to the fact that we attempt to retrieve multiple messages instead of one, it is not possible to utilize the sign assignment in  \cite{sun2018capacity} to achieve the same amount of redundancy. Instead, we propose a   new sign assignment approach resulting in this amount of redundancy while guaranteeing the decodability.
 	
 	\textbf{Step 5: Sign assignment. }  The messages in each query should be sorted based on the lexicographic order of the messages. In round $2$, for each query with one symbol from  $\{a,b,c\}$ and one symbol from  $\{d,e\}$, a plus sign is used in between the symbols; for each other query in this round, a minus sign is used. For instance, the queries in the first stage of round $2$ sent to Server $1$ would be as follows.
        \begin{align*}
            q_1=a_{25}-c_{13}, \\
            q_2=a_{26}+d_{13}, \\
            q_3=a_{27}+e_{13}, \\
            q_4=b_{25}-c_{14}, \\
            q_5=b_{26}+d_{14}, \\
            q_6=b_{27}+e_{14}, \\
            q_7=a_{14}-b_{13}, \\
            q_8=c_{26}+d_{25}, \\
            q_9=c_{27}+e_{25}, \\
            q_{10}=d_{27}-e_{26}.
        \end{align*}  
    Among these queries, the query $q_{10}=d_{27}-e_{26}$ can be written as a linear combination of the other queries, thus being redundant. To show this, suppose $d=a+b$ and $e=b+c$; one  can check that the equation $q_{10}=q_3+q_6-q_5-q_8+q_7+q_1+q_4$ holds.


  In general for the sign assignment process, each query is first divided into two parts. The first part contains symbols of independent messages and the second part symbols of dependent messages, which are called independent symbols and dependent symbols, respectively. 
  So each query $q$ is written as
 	\begin{align} \label{eq:strplus}
 		q=(\text{independent symbols})\pm(\text{dependent symbols}),
 	\end{align}
  where in each parenthesis, symbols are ordered based on the label of the message (ranging from $1$ to $M$), from lowest to highest. 
    The signs in each parenthesis are changing alternatively between $+$ and $-$, with the first symbol taking   $+$. 
    When the plus sign is used in \eqref{eq:strplus}, the sign assignment is called \textit{structure plus}; when minus sign is used, it is called \textit{structure minus}. Round $2$ queries use structure plus, then round $3$ uses minus, and round $4$ again uses plus.  This is the most non-trivial step in the proposed sign assignment. Note that  besides alternating signs in each parenthesis  in~\eqref{eq:strplus}, we also alternate the signs between the two parenthesis in each query according to the round numbers. The latter sign alternating is needed  to ensure the decodability of the scheme which will be proved in Appendix \ref{sec:decodability}. 
    After these steps, each query is solely randomly multiplied by a $+1$ or $-1$, uniformly at random, referred to as switching random variables.  This is to ensure the existence of mapping of queries for two different sets of demanded messages by a choice of $\{\sigma_i\}$ and these switching RVs, which is required in the  privacy proof in Appendix \ref{sec:privacy}. 
     We assume all to be $+1$ in the example. After the sign assignment, the queries follow Tables \ref{tab:exmplcompleter2} and \ref{tab:exmplcompleter34}. Notice that the multiplicative matrices $\mathbf{G}$ in Table \ref{tab:exmplcompleter2} is due to the next step which is explained in the following.

 	\textbf{Step 6: Remove Redundancy. } In the first stage of round~2, as pointed out in the previous step, among the queries $\{q_1,q_2,...,q_{10}\}$, $q_{10}$ is redundant.
  On the other hand, since both $a_{14}$ and $b_{13}$ are downloaded in the first and second stages of round $1$ from Server $2$, $q_7=a_{14}-b_{13}$ is also redundant. 
  However, it is not possible to simply delete this query since it jeopardizes the symmetry and consequently the privacy. Instead, we use a coding strategy as follows. Instead of sending the $10$ queries $\{q_1,q_2,...,q_{10}\}$, the following $8$ queries are sent.
        \begin{align}
            \mathbf{q} = \mathbf{G}_{8\times 10}  [q_1,q_2,...,q_{10}]^T,
        \end{align}
    where $\mathbf{G}_{8\times 10}$ is an  MDS matrix of size $8\times 10$. By receiving  all 8 queries in $\mathbf{q}$, the user is  able to decode all non-redundant queries. 
    Similarly, all queries are depicted in Tables \ref{tab:exmplcompleter1}, \ref{tab:exmplcompleter2}, and \ref{tab:exmplcompleter34}.

  \textbf{Step 7: Shuffling. } The order of queries to each server and also the order of the symbols appearing in each query are shuffled, each uniformly at random, to avoid the information leakage from the query orders and  the symbol orders. 

    \begin{remark} [Rate calculation]
        After Step $6$, there are $3$ queries in each stage of round $1$, $8$ in each stage of round $2$, $7$ in each stage of round $3$, and $2$ in each stage of round $4$, summing to the total of $184$ symbols. Since $L=68$, the proposed scheme achieves the rate $R_2=0.74$, while the baseline scheme achieves $R_1=0.61$. 
    \end{remark}

    \begin{remark} [Privacy]
        Intuitively, the privacy of the proposed scheme follows from the fact that the scheme yields symmetric queries to each server. In every stage, all possible $i$-sums appear, and   from the view point of each message, the index structure is symmetric. Besides, using the multiplicative variables $\sigma_i$, we prove in Appendix \ref{sec:privacy} 
        that the symbols signs appeared in each query have a one to one mapping for different sets of demanded messages; keeping the demanded messages hidden from the viewpoint of each server.
    \end{remark}

    \begin{remark}[Outline of the proposed scheme]
        After the initialization steps (Steps 1-3), 
        the proposed scheme  in Step~4 designs the queries similar to the delivery phase of coded caching in terms of designing the indices of symbols (whose detailed explanation will be provided in  Lemma~\ref{lem: index}). Then using the sign assignment strategy  in Step 5, 
        we let some transmitted messages be linear combinations of others, such that  
        this redundancy could be removed by using an MDS matrix in Step 6. As a result, the number of transmissions is reduced.
     Note that Step 2 is designed such that the number of side information queries needed is satisfied. 
    \end{remark}

    \begin{remark} (Why $\alpha_5=0$)
       Note that in each informative query (i.e., the query which contains  symbols  from the demanded messages), 
        there exists only one new symbol  from all the demanded messages, which has not been decoded. 
       In this example, since there are totally $5$ messages,    $2$ of which are demanded, summation of $5$ symbols (each from a different message) has two demanded symbols and cannot contribute to decoding any demanded symbols. 
       This is the reason the scheme continues till round $4$. 
    \end{remark}

    \begin{remark} (Number of stages)
        The last round is round $4$ with $1$ stage. These queries are of the form $a_*+c_*+d_*+e_*$ where $c_*+d_*+e_*$ is treated as  side information, or of the form $b_*+c_*+d_*+e_*$ where again $c_*+d_*+e_*$ is treated as  side information, or of the form $a_*+b_*+\{c_*+d_*, c_*+e_*, d_*+e_*\}$ where $\{c_*+d_*, c_*+e_*, d_*+e_*\}$ are treated as  side information. Based on this observation, a stage of round $4$ needs $2$ stages of round $3$ and $1$ stage of round $2$, to get the side information. Similarly, for a stage of round $3$, the queries are of the form $a_*+\{c_*+d_*, c_*+e_*, d_*+e_*\}$ where $\{c_*+d_*, c_*+e_*, d_*+e_*\}$ are treated as side information, or of the form $b_*+\{c_*+d_*, c_*+e_*, d_*+e_*\}$ where  $\{c_*+d_*, c_*+e_*, d_*+e_*\}$ are treated as side information, or of the form $a_*+b_*+\{c_*,d_*,e_*\}$ where $\{c_*,d_*,e_*\}$ are treated as side information. Based on this observation, a stage of round $3$ needs $2$ stages of round $2$ and $1$ stage of round $1$, to get the side information. By a similar argument, round $2$ needs $2$ stages of round $1$ to get the side information. Considering all this together, the number of stages in each round is determined as $\alpha_5=0, \alpha_4=1, \alpha_3=2, \alpha_2=5, \alpha_1=12$. For instance, for the $12$ stages of round $1$, $10$ of them are used as side information in   round $2$, which has $5$ stages and each of which needs $2$ stages of round $1$ as side information;   the remaining $2$ are used in  round $3$ with $2$ stages, since each of which needs $1$ stage of round $1$ as side information.
    \end{remark}


  \begin{table}
		\centering
		\begin{tabular}{| c | c | c | c |}
			\hline
			Round & Stage & Server $1$ & Server $2$ \\
			\hline
			\multirow{50}{*}{round $1$}
			& stage 1 
                & 
                $
 				\mathbf{G}_{3 \times 5}^{(1,1)} \times \begin{bmatrix}
 					a_1 \\
 					b_1 \\
 					c_1 \\
                        d_1 \\
 					e_1
 				\end{bmatrix}
 			$
                &
                $
 				\mathbf{G}_{3 \times 5}^{(1,1)}  \times \begin{bmatrix}
 					a_{13} \\
 					b_{13} \\
 					c_{13} \\
                        d_{13} \\
 					e_{13}
 				\end{bmatrix}
 			$ \\
			\cline{2-4}
			& stage 2 
                &
                $
 				\mathbf{G}_{3 \times 5}^{(1,2)}  \times \begin{bmatrix}
 					a_{2} \\
 					b_{2} \\
 					c_{2} \\
                        d_{2} \\
 					e_{2}
 				\end{bmatrix}
 			$
                &
                $
 				\mathbf{G}_{3 \times 5}^{(1,2)}  \times \begin{bmatrix}
 					a_{14} \\
 					b_{14} \\
 					c_{14} \\
                        d_{14} \\
 					e_{14}
 				\end{bmatrix}
 			$ \\
			\cline{2-4}
			& stage 3 
                &
                $
 				\mathbf{G}_{3 \times 5}^{(1,3)}  \times \begin{bmatrix}
 					a_{3} \\
 					b_{3} \\
 					c_{3} \\
                        d_{3} \\
 					e_{3}
 				\end{bmatrix}
 			$
                &
                $
 				\mathbf{G}_{3 \times 5}^{(1,3)}  \times \begin{bmatrix}
 					a_{15} \\
 					b_{15} \\
 					c_{15} \\
                        d_{15} \\
 					e_{15}
 				\end{bmatrix}
 			$ \\
			\cline{2-4}
			& stage 4 
                &
                $
 				\mathbf{G}_{3 \times 5}^{(1,4)}  \times \begin{bmatrix}
 					a_{4} \\
 					b_{4} \\
 					c_{4} \\
                        d_{4} \\
 					e_{4}
 				\end{bmatrix}
 			$
                &
                $
 				\mathbf{G}_{3 \times 5}^{(1,4)}  \times \begin{bmatrix}
 					a_{16} \\
 					b_{16} \\
 					c_{16} \\
                        d_{16} \\
 					e_{16}
 				\end{bmatrix}
 			$ \\
			\cline{2-4}
			& stage 5 
                &
                $
 				\mathbf{G}_{3 \times 5}^{(1,5)}  \times \begin{bmatrix}
 					a_{5} \\
 					b_{5} \\
 					c_{5} \\
                        d_{5} \\
 					e_{5}
 				\end{bmatrix}
 			$
                &
                $
 				\mathbf{G}_{3 \times 5}^{(1,5)}  \times \begin{bmatrix}
 					a_{17} \\
 					b_{17} \\
 					c_{17} \\
                        d_{17} \\
 					e_{17}
 				\end{bmatrix}
 			$ \\
			\cline{2-4}
			& stage 6 
                &
                $
 				\mathbf{G}_{3 \times 5} ^{(1,6)} \times \begin{bmatrix}
 					a_{6} \\
 					b_{6} \\
 					c_{6} \\
                        d_{6} \\
 					e_{6}
 				\end{bmatrix}
 			$
                &
                $
 				\mathbf{G}_{3 \times 5}^{(1,6)}  \times \begin{bmatrix}
 					a_{18} \\
 					b_{18} \\
 					c_{18} \\
                        d_{18} \\
 					e_{18}
 				\end{bmatrix}
 			$ \\
			\cline{2-4}
			& stage 7 
                &
                $
 				\mathbf{G}_{3 \times 5}^{(1,7)}  \times \begin{bmatrix}
 					a_{7} \\
 					b_{7} \\
 					c_{7} \\
                        d_{7} \\
 					e_{7}
 				\end{bmatrix}
 			$
                &
                $
 				\mathbf{G}_{3 \times 5}^{(1,7)}  \times \begin{bmatrix}
 					a_{19} \\
 					b_{19} \\
 					c_{19} \\
                        d_{19} \\
 					e_{19}
 				\end{bmatrix}
 			$ \\
			\cline{2-4}
			& stage 8 
                &
                $
 				\mathbf{G}_{3 \times 5}^{(1,8)}  \times \begin{bmatrix}
 					a_{8} \\
 					b_{8} \\
 					c_{8} \\
                        d_{8} \\
 					e_{8}
 				\end{bmatrix}
 			$
                &
                $
 				\mathbf{G}_{3 \times 5}^{(1,8)}  \times \begin{bmatrix}
 					a_{20} \\
 					b_{20} \\
 					c_{20} \\
                        d_{20} \\
 					e_{20}
 				\end{bmatrix}
 			$ \\
			\cline{2-4}
			& stage 9 
                &
                $
 				\mathbf{G}_{3 \times 5}^{(1,9)}  \times \begin{bmatrix}
 					a_{9} \\
 					b_{9} \\
 					c_{9} \\
                        d_{9} \\
 					e_{9}
 				\end{bmatrix}
 			$
                &
                $
 				\mathbf{G}_{3 \times 5}^{(1,9)}  \times \begin{bmatrix}
 					a_{21} \\
 					b_{21} \\
 					c_{21} \\
                        d_{21} \\
 					e_{21}
 				\end{bmatrix}
 			$ \\
			\cline{2-4}
			& stage 10 
                &
                $
 				\mathbf{G}_{3 \times 5}^{(1,10)}  \times \begin{bmatrix}
 					a_{10} \\
 					b_{10} \\
 					c_{10} \\
                        d_{10} \\
 					e_{10}
 				\end{bmatrix}
 			$
                &
                $
 				\mathbf{G}_{3 \times 5}^{(1,10)}  \times \begin{bmatrix}
 					a_{22} \\
 					b_{22} \\
 					c_{22} \\
                        d_{22} \\
 					e_{22}
 				\end{bmatrix}
 			$ \\
			\cline{2-4}
			& stage 11 
                &
                $
 				\mathbf{G}_{3 \times 5}^{(1,11)}  \times \begin{bmatrix}
 					a_{11} \\
 					b_{11} \\
 					c_{11} \\
                        d_{11} \\
 					e_{11}
 				\end{bmatrix}
 			$
                &
                $
 				\mathbf{G}_{3 \times 5}^{(1,11)}  \times \begin{bmatrix}
 					a_{23} \\
 					b_{23} \\
 					c_{23} \\
                        d_{23} \\
 					e_{23}
 				\end{bmatrix}
 			$ \\
			\cline{2-4}
			& stage 12 
                &
                $
 				\mathbf{G}_{3 \times 5}^{(1,12)}  \times \begin{bmatrix}
 					a_{12} \\
 					b_{12} \\
 					c_{12} \\
                        d_{12} \\
 					e_{12}
 				\end{bmatrix}
 			$
                &
                $
 				\mathbf{G}_{3 \times 5}^{(1,12)}  \times \begin{bmatrix}
 					a_{24} \\
 					b_{24} \\
 					c_{24} \\
                        d_{24} \\
 					e_{24}
 				\end{bmatrix}
 			$ \\
			\hline
		\end{tabular}
		\caption{Round $1$ of queries.}
  \label{tab:exmplcompleter1}
	\end{table}

 \begin{table}
		\centering
            \resizebox{\columnwidth}{!}{%
		\begin{tabular}{| c | c | c | c |}
			\hline
			Round & Stage & Server $1$ & Server $2$ \\
			\hline
			\multirow{41}{*}{round $2$}
			& \multirow{1}{*}{stage $1$}
                &
                $
 				\mathbf{G}_{8 \times 10}^{(2,1)}  \times \begin{bmatrix}
 					a_{25}-c_{13} \\
 					a_{26}+d_{13} \\
 					a_{27}+e_{13} \\
                        b_{25}-c_{14} \\
 					b_{26}+d_{14} \\
                        b_{27}+e_{14} \\
                        a_{14}-b_{13} \\
                        c_{26}+d_{25} \\
                        c_{27}+e_{25} \\
                        d_{27}-e_{26}
 				\end{bmatrix}
 			$
                &
                $
 				\mathbf{G}_{8 \times 10}^{(2,1)}  \times \begin{bmatrix}
 					a_{28}-c_{1} \\
 					a_{29}+d_{1} \\
 					a_{30}+e_{1} \\
                        b_{28}-c_{2} \\
 					a_{29}+d_{2} \\
                        b_{30}+e_{2} \\
                        a_{2}-b_{1} \\
                        c_{29}+d_{28} \\
                        c_{30}+e_{28} \\
                        d_{30}-e_{29}
 				\end{bmatrix}
 			$ \\
			\cline{2-4}
			& \multirow{1}{*}{stage $2$}
                &
                $
 				\mathbf{G}_{8 \times 10}^{(2,2)} \times \begin{bmatrix}
 					a_{31}-c_{15} \\
 					a_{32}+d_{15} \\
 					a_{33}+e_{15} \\
                        b_{31}-c_{16} \\
 					b_{32}+d_{16} \\
                        b_{33}+e_{16} \\
                        a_{16}-b_{15} \\
                        c_{32}+d_{31} \\
                        c_{33}+e_{31} \\
                        d_{33}-e_{32}
 				\end{bmatrix}
 			$
                &
                $
 				\mathbf{G}_{8 \times 10}^{(2,2)} \times \begin{bmatrix}
 					a_{34}-c_{3} \\
 					a_{35}+d_{3} \\
 					a_{36}+e_{3} \\
                        b_{34}-c_{4} \\
 					a_{35}+d_{4} \\
                        b_{36}+e_{4} \\
                        a_{4}-b_{3} \\
                        c_{35}+d_{34} \\
                        c_{36}+e_{34} \\
                        d_{36}-e_{35}
 				\end{bmatrix}
 			$ \\
			\cline{2-4}
			& \multirow{1}{*}{stage $3$}
                &
                $
 				\mathbf{G}_{8 \times 10}^{(2,3)} \times \begin{bmatrix}
 					a_{37}-c_{17} \\
 					a_{38}+d_{17} \\
 					a_{39}+e_{17} \\
                        b_{37}-c_{18} \\
 					b_{38}+d_{18} \\
                        b_{39}+e_{18} \\
                        a_{18}-b_{17} \\
                        c_{38}+d_{37} \\
                        c_{39}+e_{37} \\
                        d_{39}-e_{38}
 				\end{bmatrix}
 			$
                &
                $
 				\mathbf{G}_{8 \times 10}^{(2,3)} \times \begin{bmatrix}
 					a_{40}-c_{5} \\
 					a_{41}+d_{5} \\
 					a_{42}+e_{5} \\
                        b_{40}-c_{6} \\
 					a_{41}+d_{6} \\
                        b_{42}+e_{6} \\
                        a_{6}-b_{5} \\
                        c_{41}+d_{40} \\
                        c_{42}+e_{40} \\
                        d_{42}-e_{41}
 				\end{bmatrix}
 			$ \\
			\cline{2-4}
			& \multirow{1}{*}{stage $4$}
                &
                $
 				\mathbf{G}_{8 \times 10}^{(2,4)} \times \begin{bmatrix}
 					a_{43}-c_{19} \\
 					a_{44}+d_{19} \\
 					a_{45}+e_{19} \\
                        b_{43}-c_{20} \\
 					b_{44}+d_{20} \\
                        a_{20}+b_{19} \\
                        a_{6}-b_{5} \\
                        c_{44}+d_{43} \\
                        c_{45}+e_{43} \\
                        d_{45}-e_{44}
 				\end{bmatrix}
 			$
                &
                $
 				\mathbf{G}_{8 \times 10}^{(2,4)} \times \begin{bmatrix}
 					a_{46}-c_{7} \\
 					a_{47}+d_{7} \\
 					a_{48}+e_{7} \\
                        b_{46}-c_{8} \\
 					a_{47}+d_{8} \\
                        b_{48}+e_{8} \\
                        a_{8}-b_{7} \\
                        c_{47}+d_{46} \\
                        c_{48}+e_{46} \\
                        d_{48}-e_{47}
 				\end{bmatrix}
 			$ \\
			\cline{2-4}
			& \multirow{1}{*}{stage $5$}
                &
                $
 				\mathbf{G}_{8 \times 10}^{(2,5)} \times \begin{bmatrix}
 					a_{49}-c_{21} \\
 					a_{50}+d_{21} \\
 					a_{51}+e_{21} \\
                        b_{49}-c_{22} \\
 					b_{50}+d_{22} \\
                        b_{51}+e_{22} \\
                        a_{22}-b_{21} \\
                        c_{50}+d_{49} \\
                        c_{51}+e_{49} \\
                        d_{51}-e_{50}
 				\end{bmatrix}
 			$
                &
                $
 				\mathbf{G}_{8 \times 10}^{(2,5)} \times \begin{bmatrix}
 					a_{52}-c_{9} \\
 					a_{53}+d_{9} \\
 					a_{54}+e_{9} \\
                        b_{52}-c_{10} \\
 					a_{53}+d_{10} \\
                        b_{54}+e_{10} \\
                        a_{10}-b_{9} \\
                        c_{53}+d_{52} \\
                        c_{54}+e_{52} \\
                        d_{54}-e_{53}
 				\end{bmatrix}
 			$ \\
			\hline
		\end{tabular} }
		\caption{Round 2 of queries.}
  \label{tab:exmplcompleter2}
	\end{table}

        \begin{table} 
		\centering
            \resizebox{\columnwidth}{!}{%
		\begin{tabular}{| c | c | c | c |}
			\hline
                Round & Stage & Server $1$ & Server $2$ \\
                \hline
			\multirow{11}{*}{round $3$}
			& \multirow{1}{*}{stage $1$}
			& $\mathbf{G}_{7 \times 10}^{(3,1)} \times \begin{bmatrix}
                    a_{55}-c_{29}-d_{28} \\
			      a_{56}-c_{30}-e_{28} \\
			      a_{57}-d_{30}+e_{29} \\
			      b_{55}-c_{35}-d_{34} \\
			    b_{56}-c_{36}-e_{34} \\
			    b_{57}-d_{36}+e_{35} \\
			        a_{34}-b_{28}+c_{23} \\
			    a_{35}-b_{29}-d_{23} \\
			    a_{36}-b_{30}-e_{23} \\
			    c_{57}-d_{56}+e_{55}
                    \end{bmatrix} $
                &
                $\mathbf{G}_{7 \times 10}^{(3,1)} \times \begin{bmatrix}
                    a_{58}-c_{26}-d_{25} \\
			      a_{59}-c_{27}-e_{25} \\
			      a_{60}-d_{27}+e_{26} \\
			      b_{58}-c_{32}-d_{31} \\
			    b_{59}-c_{33}-e_{31} \\
			    b_{60}-d_{33}+e_{32} \\
			        a_{31}-b_{25}+c_{11} \\
			    a_{32}-b_{26}-d_{11} \\
			    a_{33}-b_{27}-e_{11} \\
			    c_{60}-d_{59}+e_{58}
                    \end{bmatrix} $ \\
			\cline{2-4}
			& \multirow{1}{*}{stage $2$}
			& $\mathbf{G}_{7 \times 10}^{(3,2)} \times \begin{bmatrix}
                  a_{61}-c_{41}-d_{40} \\
			     a_{62}-c_{42}-e_{40} \\
			   a_{63}-d_{42}+e_{41} \\
			   b_{61}-c_{47}-d_{46} \\
			   b_{62}-c_{48}-e_{46} \\
			   b_{63}-d_{48}+e_{47} \\
			   a_{46}-b_{40}+c_{24} \\
			   a_{47}-b_{41}-d_{24} \\
			   a_{48}-b_{42}-e_{24} \\
			   c_{63}-d_{62}+e_{61}
                  \end{bmatrix} $
                &
                $\mathbf{G}_{7 \times 10}^{(3,2)} \times \begin{bmatrix}
                  a_{64}-c_{38}-d_{37} \\
			     a_{65}-c_{39}-e_{37} \\
			   a_{66}-d_{39}+e_{38} \\
			   b_{64}-c_{44}-d_{43} \\
			   b_{65}-c_{45}-e_{43} \\
			   b_{66}-d_{45}+e_{44} \\
			   a_{43}-b_{37}+c_{12} \\
			   a_{44}-b_{38}-d_{12} \\
			   a_{45}-b_{39}-e_{12} \\
			   c_{66}-d_{65}+e_{64}
                  \end{bmatrix} $ \\
			\hline
			\multirow{1}{*}{round $4$}
			& \multirow{1}{*}{stage $1$}
			&   $\mathbf{G}_{2 \times 5}^{(3,2)} \times \begin{bmatrix}
                    a_{67}-c_{60}+d_{59}-e_{58} \\
			    b_{67}-c_{66}+d_{65}-e_{64} \\
			    a_{64}-b_{58}+c_{53}+d_{52} \\
			    a_{65}-b_{59}+c_{54}+e_{52} \\
			    a_{66}-b_{60}+d_{54}-e_{53} 	
                    \end{bmatrix} $
                &
                    $\mathbf{G}_{2 \times 5}^{(3,2)} \times \begin{bmatrix}
                    a_{68}-c_{57}+d_{56}-e_{55} \\
			    b_{68}-c_{63}+d_{62}-e_{61} \\
			    a_{61}-b_{55}+c_{50}+d_{49} \\
			    a_{62}-b_{56}+c_{51}+e_{49} \\
			    a_{63}-b_{57}+d_{51}-e_{50}
                    \end{bmatrix} $ \\
			\hline			
		\end{tabular} }
	\caption{Rounds 3 and 4 of queries.}
 \label{tab:exmplcompleter34} 
\end{table}

	\section{New Proposed MM-PC Scheme: The General Case} \label{sec:achievable}
       In this section, following the main idea of the example in Section~\ref{sec:exp achievable}, we  describe the general  MM-PC scheme proposed in this paper. Note that each message is divided into $L$ symbols. The $j^{\text{th}}$ symbol of $W_i$ is denoted by $W_i(j)$. The proofs of decodability and privacy of the proposed scheme are provided in Appendices \ref{sec:decodability} and \ref{sec:privacy}, respectively. 

        \textbf{Step 1: Permutation and Relabeling. } In this step, the symbols in each message are permuted by a single permutation function $\pi(\cdot)$ over $[L]$ and multiplied by the multiplicative variable $\sigma_i \in \{+1,-1\}$ for the symbol index $i\in [L]$. We denote the alternated message of $W_m$ by $u_m$ as follows.
	\begin{align}
		u_m(i) := \sigma_i W_m(\pi(i)), m \in [M], i \in [L].
  \label{eq:i the symbol +-1}
	\end{align}
	Both the permutation function $\pi$ and the multiplicative variables $\sigma_i$ are uniformly and independently distributed. Notice that these functions are independent of message label $m\in [M]$.

        Furthermore, we change the initial labeling of the messages such that the first $P$ labels are the demanded messages; i.e., $(\theta_1,\theta_2,...,\theta_P)=(1,2,...,P)$. We expand the new basis with $K-P$ more independent messages with the new labels from $P+1$ to $K$, and then label the others (which are the new dependent ones) from $K+1$ to $M$. Notice that this is possible with the assumption that the demanded messages are independent. This relabeling (or permutation on messages) is done privately by the user and unknown to the servers.

	\textbf{Step 2: Number of Stages.} The main idea of this step  is inspired from the second MM-PIR scheme in~\cite{banawan2018multi} (i.e., for case  $P \leq \frac{M}{2}$).  
    The query structure to each server is split into $M-P+1$ \textit{rounds}, where each  round $i$ contains the queries summing $i$ different symbols. Each round may also be split into multiple \textit{stages}. Each stage of round $i$ queries contains all $\binom{M}{i}$ possible choices of messages; i.e., summations with the form $u_{j_1}(*)+u_{j_2}(*)+...+u_{j_i}(*), \forall \{j_1,j_2,...,j_i\} \subset [M]$. The symbol indices $*$ will be carefully chosen, explained in the index assignment step. In each round, the number of stages will be determined as follows. 
    Consider a stage of round $i$ queries to server $1$. The queries are partitioned based on the number of symbols from the demanded messages involved. For the queries containing only $1$ symbol from the demanded messages,  there are $\binom{P}{1}=P$ types; for each type, one stage of round $i-1$ is needed to provide the side information part. Note that these $P$ stages of round $i-1$ (used for providing side information) are from the other $N-1$ servers, i.e., servers $2$ to $N$, for the sake of privacy. 
   Generally,  for the queries containing $i_1 \in [\min \{i, P\}]$ symbols from the demanded messages,  there are $\binom{P}{i_1}$ types; for each type, one stage of round $i-i_1$ is needed from the other $N-1$ servers to provide the side information part. 

    The number of stages in each round $j$ queries to each server is denoted by $\alpha_j$, for $j\in [M-P+1]$. The number of  stages of round $j$ to servers $2$ to $N$ would be $(N-1)\alpha_j$. $\binom{P}{1} \alpha_{j+1}$ of these stages will be  used as the side information in $\alpha_{j+1}$ stages of round $j+1$ queries to server $1$. $\binom{P}{2} \alpha_{j+2}$ of these stages will be used as the side information in $\alpha_{j+2}$ stages of round $j+2$ queries to server $1$, and so on, leading to the equation \eqref{eq:stages}. Furthermore, as  seen in the example, only queries containing one symbol from the demanded messages contribute to decoding new demanded symbols. 
    Thus, after round $M-P+1$, since each query would have at least two symbols from the demanded messages, the scheme is designed to continue only until round $M-P+1$; $\alpha_j=0, \forall j \in [M-P+2:M]$. 

	After the general structure of the queries is set, the next step would be to determine which indices should be used for the symbols in each query. 
 
	\textbf{Step 3: Initialization. } In this step, the queries of round~$1$ (single symbols) are downloaded from the servers. Let \texttt{new}($u_m$) be a function that starting from $u_m(1)$, returns the next symbol index of $u_m$ each time it is called, i.e., the first time the function \texttt{new}($u_m$) is called, it returns $u_m(1)$, next time it returns $u_m(2)$ and so on. Starting from server~$1$, the functions $\texttt{new}(u_1), \dots, \texttt{new}(u_M)$ are called as the queries to the server. This is one stage of round $1$ and should be repeated $\alpha_1$ times in total for each server. 
	
	
	
	To determine the indices of queries in rounds $2$ to $M-P+1$, we have the following step. Notice that the ultimate goal of  index assignment, is to exploit the redundancy between messages and reduce the total number of queries, hence increasing the rate. 
	
	\textbf{Step 4: Index assignment. } 
 The indexing structure follows the following lemma.
\begin{lemma}[Index structure]
\label{lem: index}
    In a stage in round $i$, for  any set of $i-1$ messages (assumed to be $\{u_1, u_2, \dots, u_{i-1}\}$) and  any other two messages (assumed to be $u_{i_1}, u_{i_2}$), in the queries with the form   $u_{i_1}(k_1)+u_1(*)+\dots +u_{i-1}(*)$ and  $u_{i_2}(k_2)+u_1(*)+\dots +u_{i-1}(*)$, it should have  $k_1=k_2$ (i.e., the symbol indices of $u_{i_1}, u_{i_2}$ in the two queries are the same.). 
\end{lemma}
 Note that the indexing structure in Lemma~\ref{lem: index} is inspired from  the delivery phase of the seminal coded caching scheme in~\cite{maddah2014fundamental}.

      Our objective is to design the index assignment satisfying Lemma \ref{lem: index}.  
    To accomplish this, we  divide the queries in each stage into three groups: (1) informative queries,  
      (2) side information queries, and (3) useless queries. The description of the design comes as follows. 
	
	\textbf{4.1: Informative Queries. } 
    These queries are used to decode new symbols of demanded messages. Each informative query only contains one symbol from demanded messages, which is added to some side information obtained from the previous round, and thus can be decoded using this side information. Formally, these queries for round $i$ are with the form $q_{\mathcal{\gamma}}=u_{\theta}(*)+u_{j_1}(*)+...+u_{j_{i-1}}(*)$ where $\theta \in [P]$ and $\{j_1,...,j_{i-1}\}\subseteq  [M]\setminus [P]$  and $\gamma$ denotes the set of message indices, i.e. $\gamma=\{\theta, j_1,...,j_{i-1}\}$. The part $u_{j_1}(*)+...+u_{j_{i-1}}(*)$   is treated as  side information directly obtained from some stage in round $i-1$  dedicated for the usage of side information for $u_{\theta}$. The symbol $u_{\theta}(*)$ is a previously not-decoded symbol for $u_{\theta}$, i.e. $\texttt{new}(u_{\theta})$.  So the queries involving $u_{\theta}$ is the set $\{q_{\theta \cup \gamma'}=\texttt{new}(u_{\theta})+u_{j_1}(*)+...+u_{j_{i-1}}(*): \forall \gamma'=\{j_1,...,j_{i-1}\}\subseteq  [M]\setminus [P]\}$ where we choose $\gamma'$ in a lexicographic order.  
    By the structure of queries, in each stage of round $i$, $\binom{M-P}{i-1}$ new symbols of each demanded message is decoded,   which equals the number of ways of choosing the set $\{j_1,...,j_{i-1}\}\subseteq  [M]\setminus [P]$. 
      Note that for any given $\gamma'=\{j_1,...,j_{i-1}\}\in [M]\setminus [P]$, in the set of queries $\{q_{\theta \cup \gamma'}=u_{\theta}(*)+u_{j_1}(*)+...+u_{j_{i-1}}(*): \forall \theta \in [P]\}$, all $u_{\theta}(*)$ where $\theta \in [P]$ have the same index, since for each $u_{\theta}$ the queries have been built on the lexicographic order of $\gamma'$, consequently satisfying the index structure in Lemma \ref{lem: index}.  
  
  Let us go back to  the  example  in Section~\ref{sec:exp achievable}. In one stage of round $2$, we first determine the query $a_*+c_*$, then $a_*+d_*$, and then $a_*+e_*$.  These queries would be $a_{25}+c_{13}$, $a_{26}+d_{13}$, and $a_{27}+e_{13}$ for the first stage of round $2$ queries to server $1$, where $a_{25}, a_{26}, a_{27}$ are new symbols of message $a$, and $c_{13}, d_{13}, e_{13}$ have been  downloaded symbols in round~$1$ treated as the side information in round $2$.


    \textbf{4.2: Side Information Queries. } 
    These queries do not contain any symbols from demanded messages. Consider the query $q_{\gamma}=u_{j_1}(*)+\dots+u_{j_i}(*)$ in round $i$ where $\gamma=\{j_1,...,j_{i}\} \in [M] \setminus [P]$. To determine the symbol index for $u_{k}$ where $k\in \gamma$, by Lemma \ref{lem: index}, 
    the index of this symbol should be determined by any informative query  
in the same stage (determined in Step $4.1$) containing symbols of messages $\gamma \setminus \{k\}$, $q_{\theta \cup \gamma \setminus \{k\}}$ for any $\theta \in [P]$; i.e., the index should be the same as the symbol index of the  demanded message $u_{\theta}$ in $q_{\theta \cup \gamma \setminus \{k\}}$.
 By the definition, the number of the side information queries in a stage in round $i$ is $\binom{M-P}{i}$.
 
 Let us go back to  the  example  in Section~\ref{sec:exp achievable}. In the first stage of round~$2$, 
 the symbol indices in $c_*+d_*$  are determined based on the informative queries $a_{25}+c_{13}$ and $a_{26}+d_{13}$. So  $c$ is added to a symbol with index $25$ and $d$ is added to a symbol with index $26$; i.e.,  the  resulting query is  $c_{26}+d_{25}$.  

	\textbf{4.3: Useless Queries. } 
    These queries contain more than one symbol from the demanded messages. Starting with the queries containing two demanded messages,  consider the query $q_{\gamma}=u_{\theta_1}(*)+u_{\theta_2}(*)+u_{j_1}(*)+...+u_{j_{i-2}}(*)$ where $\gamma=\{\theta_1, \theta_2, j_1, \dots, j_{i-2}\}$, and $\theta_1,\theta_2 \in [P], j_1,...j_{i-2} \in [M] \setminus [P]$. The part $u_{j_1}(*)+...+u_{j_{i-2}}(*)$ is a side information obtained from a stage in round $i-2$. Therefore, it  remains to determine the indices of $u_{\theta_1}$ and $u_{\theta_2}$. To determine the index of $u_{\theta_1}$, based on Lemma \ref{lem: index}, 
the index of this symbol should be determined by any informative query  
in the same stage (determined in Step $4.1$) containing symbols of messages $\{u_{\theta_2},u_{j_1},...,u_{j_{i-2}}\}$, $q_{\gamma'}$ where $\gamma'=\theta'_1 \cup \{\theta_2, j_1, \dots, j_{i-2}\}$ for any $\theta'_1 \in [M]\setminus [P]$; i.e., the index should be the same as the symbol index of $u_{\theta_1'}(*)$ in $q_{\gamma'}$.
It is important to note that since $u_{\theta_1'}(*)$ comes from a side information query in a stage of round $i-1$,   $u_{\theta_1}$ with the same symbol index has also appeared in the same stage in an informative query, and thus   has already been decoded there.  
    Consequently, these queries cannot contribute to decoding new symbols for demanded messages, nor serve as side information. 
    We observe that to determine symbol indices containing two symbols from the demanded messages, queries containing one are used. Similarly, to determine symbol indices for queries containing three symbols from demanded messages, queries containing two are used, with a similar process explained. This process continues until all queries in this group have been indexed. 

    Let us go back to  the  example  in Section~\ref{sec:exp achievable}.
    In the first stage of round~$1$,    to determine the indices in $a_*+b_*$, we need to check the queries $a_{25}+c_{13}$ and $b_{25}+c_{14}$. So $a$ is added to a symbol with index $13$ and $b$ to a symbol with index $14$; i.e., the resulting query is  $a_{14}+b_{13}$.

 	As a result, by Step 4, the indices of all the symbols are determined. The next step would be to assign the signs ($+1$ or $-1$) to symbols in the queries, such that there would be some queries being linear combinations of other queries.
  
 	
 	\textbf{Step 5: Sign assignment. }  The sign assignment step, from round $2$ to the last round, includes two sub-steps: (1) choosing between structure plus or minus and (2) performing random sign switching, which are described as follows. 

  \textbf{5.1: Structure Plus/Minus.} Each query is first divided into two parts. The first part contains symbols from independent messages and the second part symbols from dependent messages. So each query $q$ is written as
 	\begin{align} \label{eq:strplus1}
 		q=(\text{independent symbols})\pm(\text{dependent symbols}).
 	\end{align}
 	The sign $+$ is   referred to as \textit{structure plus} and the sign $-$ is referred to as \textit{structure minus}. In round $2$, a structure plus is used in each query. The structure is successively switched for the next rounds, i.e. for round $3$, a structure minus is used in each query; for round $4$, a structure plus is used in each query; and so on. Additionally, in each parenthesis, after ordering the symbols based on the lexicographic order of the corresponding messages, the first symbol is assigned by a plus sign and this successively alternates until the last symbol in the parenthesis. In other words, if the independent symbols in~\eqref{eq:strplus1} are $u_{i_1}(*), u_{i_2}(*),\ldots, u_{i_j}(*)$ where $i_1<i_2<\cdots<i_j$, then $(\text{independent symbols})$ in~\eqref{eq:strplus1} should be 
  \begin{align}
    (\text{independent symbols})=(u_{i_1}(*)-u_{i_2}(*)+u_{i_3}(*)-\cdots).  
  \end{align}
  Similarly, if the dependent symbols in~\eqref{eq:strplus1} are $u_{k_1}(*), u_{k_2}(*),\ldots, u_{k_j}(*)$ where $k_1<k_2<\cdots<k_j$, then $(\text{dependent symbols})$ in~\eqref{eq:strplus1} should be 
  \begin{align}
    (\text{dependent symbols})=(u_{k_1}(*)-u_{k_2}(*)+u_{k_3}(*)-\cdots).  
  \end{align}
  

  \textbf{5.2: Random Sign Switching:} In this step,  each query solely is multiplied by $+1$ or $-1$, uniformly and independently at random.

  \begin{remark}
      As studied in \cite{wan2021optimal} in the cache-aided scalar linear function retrieval problem, in order to reduce the load in the delivery phase of a caching system in which each user requests a linear combination of messages, it is needed that symbols get multiplied by a minus or a plus based on certain rules. For sign assignment,  we are inspired from the sign assignment in~\cite{wan2021optimal}. Particularly, the caching scheme in~\cite{wan2021optimal} always uses the  structure plus between independent symbols and dependent symbols.
      This is natural since they have one stage (and also only one round) of delivery. However, since we have multiple delivery stages and rounds, which are are inter-connected; i.e. a side information query in one stage is used in another stage, to ensure the decodability of the scheme, we have to use the plus and minus structures alternatively in rounds. 
  \end{remark}



 	 \begin{lemma} \label{lem:redundancy}
 		By the end of Step 5, each stage of round $i$ has $\binom{M-K}{i}$  linearly redundant queries from the total $\binom{M}{i}$ queries,  and can be written as linear combinations of the others. Linearly redundant queries are those which do not contain any symbols from independent messages.
 	\end{lemma}
 The proof of this theorem is given in Appendix~\ref{sec:proof of redundancy}.
 
 	By our construction  up to the end of Step 5, 
  it is important to summarize that in each stage there are two disjoint sets of redundant queries: the set of useless queries and the set of linearly redundant queries. More precisely,
 	\begin{itemize}
 		\item The useless queries are  redundant since they are the summation of some side information and some symbols of demanded messages which are all  previously decoded.
 		\item 
   The set of linearly redundant queries by Lemma \ref{lem:redundancy} are among the side information queries, which are some linear combinations of all remaining queries.   
 	\end{itemize}
  Hence, we can further reduce the amount of download summations by removing the redundancy. 
  However, removing these queries  directly from the set of queries jeopardizes privacy. Step 6 introduces a way to  reduce download  while preserving privacy.
 	
 	For a stage of round $i \in [M-P+1]$, the number of informative queries, side information queries, and useless queries are $n^{(i)}_{iq}=P \binom{M-P}{i-1}, n^{(i)}_{sq}=\binom{M-P}{i}, n^{(i)}_{uq}=\binom{M}{i} - n^{(i)}_{iq} - n^{(i)}_{sq}$, respectively. The number of linearly redundant queries are $n^{(i)}_{rq} = \binom{M-K}{i}$. 
 	
 	\textbf{Step 6: Reducing Download. } For each round $i$ and each stage $s$, if the queries in state $s$ of round $i$ are $q_{1},\ldots,q_{\binom{M}{i}}$, we denote $\mathbf{q}^{(i,s)}=\left[ \begin{array}{c}
q_1\\
q_2\\
\vdots\\
q_{\binom{M}{i}}
\end{array} \right]$. 
 We multiply $\mathbf{q}^{(i,s)}$ on the left by the MDS matrix $\mathbf{G}^{(i,s)}$ of size $r \times \binom{M}{i}$, where $r$ is defined as $r=\binom{M}{i}-n^{(i)}_{uq}-n^{(i)}_{rq}=P\binom{M-P}{i-1}+\binom{M-P}{i}-\binom{M-K}{i}$, to reach the final set of queries in this stage as the elements of $\mathbf{q}^{(i,s)}_f$,
	 \begin{align} \label{eq:redundancy}
		\mathbf{q}^{(i,s)}_f := \mathbf{G}_{r \times \binom{M}{i}}^{(i,s)} \mathbf{q}^{(i,s)}.
	\end{align}
	This is done for all rounds $i$ and stages $s$.
	
	The reason we can decode  all $\binom{M}{i}$ queries in $\mathbf{q}^{(i,s)}$ by $\mathbf{q}^{(i,s)}_f$ is as follows. We first partition $\mathbf{q}^{(i,s)}$ into three parts as
	\begin{align}
		\mathbf{q}^{(i,s)} = \begin{bmatrix}
												\mathbf{q}^{(i,s)}_{1} \\
												\mathbf{q}^{(i,s)}_{2} \\
												\mathbf{q}^{(i,s)}_{3}
											\end{bmatrix},
	\end{align}
	where $\mathbf{q}^{(i,s)}_{3}, \mathbf{q}^{(i,s)}_{2}, \mathbf{q}^{(i,s)}_{1}$ represent linearly redundant queries, useless queries, and other queries, respectively. Since $\mathbf{q}^{(i,s)}_{3}$ is a linear combination of the other two, there exists a full rank matrix $\mathbf{G}'$ such that
	\begin{align}
		\mathbf{q}^{(i,s)} = \begin{bmatrix}
			\mathbf{q}^{(i,s)}_{1} \\
			\mathbf{q}^{(i,s)}_{2} \\
			\mathbf{q}^{(i,s)}_{3}
		\end{bmatrix} = \mathbf{G}' \begin{bmatrix}
		\mathbf{q}^{(i,s)}_{1} \\
		\mathbf{q}^{(i,s)}_{2}
	\end{bmatrix}.
	\end{align}
 	 Thus, \eqref{eq:redundancy} turns into
 	\begin{align}
 		\mathbf{q}^{(i,s)}_f  = \mathbf{G}^* \begin{bmatrix}
 			\mathbf{q}^{(i,s)}_{1} \\
 			\mathbf{q}^{(i,s)}_{2}
 		\end{bmatrix},
 	\end{align}
	for some full rank matrix $\mathbf{G}^*=\mathbf{G}_{r \times \binom{M}{i}}^{(i,s)} \mathbf{G}'$. Since the queries in $\mathbf{q}^{(i,s)}_{2}$  have already been decoded from the previous rounds, together with $\mathbf{q}^{(i,s)}_f$ we can decode $\mathbf{q}^{(i,s)}_{1}$. 


  \textbf{Step 7: Shuffling. } Finally, we shuffle the order of queries sent to each server and also, shuffle the order of the messages appearing in each query. The shufflings are uniformly and independently at random. This is to prevent servers from guessing any orders between messages and queries.

\textbf{Decodability and rate.} Intuitively,  the decodability simply follows since the informative queries  are composed of the desired symbol added to some previously downloaded side information; the most-non-trivial step to guarantee this is the alternative  
 structure plus and structure minus cross different rounds. 
The overall rate is computed as the ratio of the number of informative queries to all queries. The formal proof of the decodability and rate computation is given in Appendix~\ref{sec:decodability}.

\textbf{Privacy.} Intuitively, privacy is satisfied since the queries are symmetric with respect to each message through the index assignment structure. Besides,  the sign assignment step does not reveal the identity of the demanded messages since there is a mapping of symbol signs for different demand scenarios with the help of random variables involved, including the multiplicative factors in Step 1 and sign switching variables in Step 5.2. As a consequence, all possible symbol signs for different demand scenarios will be equally likely. Furthermore, the MDS coding step trivially does not jeopardise privacy. The formal proof of the privacy is given in Appendix~\ref{sec:privacy}.
 
        \section{Conclusion}
    In this paper, we studied the multi-message private computation problem which is an extension to the PC problem of \cite{sun2018capacity} and the MM-PIR of \cite{banawan2018multi}. Our design is based on breaking the scheme into multiple rounds and stages such that round $i$ corresponds to queries in the form of summations of $i$ different symbols. By designing the index and sign of each symbol involved, we were able to reduce the amount of downloaded summations since some of the queries are linear combinations of the others. Furthermore, to use this redundancy while preserving privacy, we used an MDS coding method so that each server cannot distinguish between the redundant and non-redundant queries. Numerical evaluations demonstrated that the rate of the proposed scheme has significant improvements over the baseline scheme for a wide range of system parameters,  thus inheriting the order-optimality of the baseline scheme within a multiplicative factor of $2$. It is also important to point out that the rate of the proposed scheme has very little dependence on $M$, as suggested by Figure \ref{fig:compare_msweep}, while this is not the case for the baseline scheme. This is important since we expect that as long as $K$ is fixed, changing only the number of possible linear combinations should not affect the rate for an order optimal scheme. We observe the same behaviour for the optimal PC scheme in \cite{sun2018capacity}.

    On-going works include deriving the converse bound specifically for the MMPC problem and designing new MMPC schemes with low subpacketization level. 

 \appendices


 	\section{Proof of decodability and rate calculation for the MM-PC scheme} \label{sec:decodability}
        Up until the end of Step 4 (index assignment), it is straightforward to decode the new symbols of demanded messages. This is because these new symbols only exist in informative queries which are built by the addition of these symbols to some already known side information. But after Step 5 (sign assignment), some symbol signs alter to a minus. Since in each stage, the informative and useless queries are build up using some side information from earlier rounds, we should check if after the sign assignment step, these side information queries remain consistent regarding the symbol signs. Before we continue, for the sake of simplicity, we assume other than the first $P$ labels, the other $K-P$ independent messages are labeled from from $P+1$ to $K$. Also for the sake of simplicity, we denote symbols just by the message letter and not using $(*)$ in front of it.
        
        In round $i$, for some informative query $q=u_{\theta}+q_{si}$ where $\theta \in [P]$, the side information part $q_{si}$ within this query should remain consistent on symbol signs compared to the corresponding query in round $i-1$ after sign assignment. Without loss of generality, assume we use structure plus for round $i-1$ and structure minus for round $i$. Also assume from the $i-1$ symbols in $q_{si}$, $v$ of them are symbols of independent messages; i.e., $q_{si}=u_{j_1}+...+u_{j_v}+u_{j_{v+1}}+...+u_{j_{i-1}}$, where $\{j_1,...j_v\} \subset [P+1:K], j_{v+1},...j_{i-1} \subset [K+1:M]$. If $v$ is even, then after sign assignment for query $q_{si}$ in round $i-1$, $u_{j_1}$ would have a plus sign and $u_{j_2}$ a minus sign and so on, until a minus sign for $u_{j_v}$. Since structure plus is used for this round, $u_{j_{v+1}}$ starts with a plus sign and the other signs follow the alternating structure; leading to $q'_{si}=u_{j_1}-u_{j_2}+...-u_{j_v}+u_{j_{v+1}}-u_{j_{v+2}}+...\pm u_{j_{i-1}}$, where $q'_{si}$ is $q_{si}$ after sign assignment. In sign assignment for the query $q$ in round $i$, $u_{\theta}$ starts with a plus sign, $u_{j_1}$ would have a minus sign, $u_{j_2}$ a plus sign up until $u_{j_v}$ with a plus sign. Then, since structure minus is used in this round, $u_{j_{v+1}}$ would start with a minus sign and so on; leading to $q'=u_{\theta}-u_{j_1}+u_{j_2}-...+u_{j_v}-u_{j_{v+1}}+u_{j_{v+2}}-...\pm u_{j_{i-1}}$, where $q'$ is $q$ after sign assignment. It is evident that $q'=u_{\theta}-q'_{si}$, and therefore, the signs are consistent after sign assignment and $q'_{si}$ can be cancelled out to decode for $u_{\theta}$. We can similarly prove the case for $v$ being odd. This completes the proof of consistency for informative queries. 

        We should prove the consistency for useless queries too. Consider the useless query $q=u_{\theta_{l_1}}+...+u_{\theta_{l_n}}+q_{si}$ in round $i$ with $n$ symbols from the demanded messages, i.e. $\{\theta_{l_1}, ..., \theta_{l_n}\} \subset [P]$ and the side information part has $v$ symbols from demanded messages, i.e. $q_{si}=u_{j_1}+...+u_{j_v}+u_{j_{v+1}}+...+u_{j_{i-n}}$, where $\{j_1,...j_v\} \in [P+1:K], j_{v+1},...j_{i-n} \in [K+1:M]$. Assume without loss of generality, in round $i-n$ structure plus is used for sign assignment. For the case $v$ is odd, after sign assignment for query $q_{si}$, $u_{j_1}$ would have a plus sign, $u_{j_2}$ a minus and so on, until $u_{j_v}$ with a plus sign. $u_{j_{v+1}}$ would have a plus sign and the rest change their signs alternatively, leading to $q'_{si}=u_{j_1}-...+u_{j_v}+u_{j_{v+1}}-...\pm u_{j_{i-n}}$, where $q'_{si}$ is $q_{si}$ after sign assignment. There are two cases for $n$, both of which need to be checked. For the case $n$ is even, for round $i$ structure plus will be used again. After sign assignment for the query $q$, $u_{\theta_{l_1}}$ would have a plus sign, $u_{\theta_{l_2}}$ a minus sign and so on, up to $u_{\theta_{l_n}}$ with a minus sign. Also, $u_{j_1}$ would have a plus, $u_{j_2}$ a minus, up until $u_{j_v}$ with a plus. Furthermore, $u_{j_{v+1}}$ would have a plus sign and the rest change their signs alternatively, leading to $q'=u_{\theta_{l_1}}-...-u_{\theta_{l_n}}+u_{j_1}-...+u_{j_v}+u_{j_{v+1}}-...\pm u_{j_{i-n}}$, where $q'$ is $q$ after sign assignment. Thus, it is evident that $q'=u_{\theta_{l_1}}-u_{\theta_{l_2}}+...-u_{\theta_{l_n}}+q'_{si}$. Therefore, again the signs remain consistent after sign assignment. For the case $n$ is odd, similarly it will be resulted that $q'=u_{\theta_{l_1}}-u_{\theta_{l_2}}+...+u_{\theta_{l_n}}-q'_{si}$, where again the consistency is evident. For the case $v$ is even, one can verify the sign consistency similarly. Therefore, we have proved the consistency of signs after the sign assignment step. Notice that in the proof, for convenience, we have assumed the sign switching variables in Step 5.2 are all $1$ and this does not jeopardize the generality, since only the relative symbol signs are important. This completes the proof of decodability.
        
 	To calculate the rate, we first need to calculate the length of the messages $L$. To do so, we need to know the number of informative queries corresponding to each demanded message, since these are the queries that generate new indices in the scheme. This has already been calculated in Step $4.1$ of the scheme as $\binom{M-P}{i-1}$, for a stage in round $i$. Furthermore, the number of stages in round $i$ is determined by $\alpha_i$ which follows~\eqref{eq:stages}. Therefore, collectively from all servers, for each message, $N \alpha_i \binom{M-P}{i-1}$ new symbols appear in round $i$. So the message length is
 	
 	\begin{align} \label{eq:messagelength}
 		L = N \sum_{i=1}^{M-P+1} \alpha_i \binom{M-P}{i-1}.
 	\end{align}
 	
 	Next, we need to calculate the total download $D$ from all servers. Based on Step $6$ of the scheme, in a stage in round $i$, a total of $r=\binom{M-P}{i} - \binom{M-K}{i} + \binom{P}{1} \binom{M-P}{i-1}$ symbols is downloaded. Considering all stages and all servers,
 	
 	\begin{align} \label{eq:download}
 		&D = \nonumber \\
 		&N \sum_{i=1}^{M-P+1} \alpha_i \left(\binom{M-P}{i} - \binom{M-K}{i} + P \binom{M-P}{i-1}\right).
 	\end{align}  
 	
 	For the rate defined in~\eqref{eq:rate}, using~\eqref{eq:messagelength} and~\eqref{eq:download}, we get
 	
 	\begin{align}
 		R_2 = \frac{P \sum_{i=1}^{M-P+1} \alpha_i \binom{M-P}{i-1}}{\sum_{i=1}^{M-P+1} \alpha_i \left(\binom{M-P}{i} - \binom{M-K}{i} + P \binom{M-P}{i-1}\right)}.
 	\end{align}
 	
 	\section{Proof of privacy for the MM-PC scheme} \label{sec:privacy}

        To prove privacy, we must show no matter the choice of $\mathcal{I}$, the realization of the queries for each server has the same probability space.
        We first point out that by the end of index assignment step, the queries to each server are completely symmetrical. This is because the queries are partitioned to multiple stages, and in each stage of round $i$, all the possible $\binom{M}{i}$ types of queries appear. Besides, the indexing structure is also symmetrical from the viewpoint of each message. As stated in Lemma \ref{lem: index}, the index structure has the following general rule: In a stage in round $i$, choose any $i-1$ messages. The set of queries with symbols of these messages have the same index for the other symbol involved in the query.


        To proceed with the proof, we first state the following lemma. 

        \begin{lemma} \label{lemma:disjointindices}
            In a stage of queries to one server, the symbol indices appearing are disjoint from those of other stages in the same server.
        \end{lemma}

        \begin{proof}
            We go through all $3$ types of queries in a stage. We first point out that since the side information queries to a server duplicate the new symbol indices of demanded messages in the same stage, and since these new indices do not appear in the same server in any other stage by definition, these queries have completely disjoint indices compared to other stages in the same server. Furthermore, the side information parts of informative and useless queries have also disjoint indices, since these parts are duplicated from queries to other servers and are used only once in queries to each server, so they do not appear anywhere else in the same server. Additionally, the symbols of demanded messages in useless queries duplicate the new indices of demanded messages in the same server, indicating they do not appear twice in queries to the same server.
        \end{proof}

        With Lemma \ref{lemma:disjointindices} and the symmetry of indices from the perspective of each message, it is readily concluded that for any two choices of demanded messages $\mathcal{I}_1$ and $\mathcal{I}_2$ where $\mathcal{I}_1 \neq \mathcal{I}_2$, the indices of symbols in queries to one server have a one to one mapping by a choice of permutation function $\pi$.


    The proposed scheme has two permutations: one on symbol indices and the other on message indices, where the latter is referred to as \textit{relabeling} as stated in the first step of the scheme. So far we have shown that the permutation function $\pi$ on symbol indices preserves privacy.
    To complete the proof, it only remains to show that the sign assignment step does not jeopardize the symmetry of the queries, in the sense that it does not reveal the private relabeling of the messages, otherwise some information on the requested messages would be leaked. We indicate this by showing that the signs of symbols in queries to one server for two choices of demanded messages $\mathcal{I}_1$ and $\mathcal{I}_2$ where $\mathcal{I}_1 \neq \mathcal{I}_2$, have an one to one mapping by a particular choice of multiplicative variables $\sigma_i, i \in [L]$ and sign switching variables in Step 5.2. Remember that these variables are chosen by the user and private to the server.

    We now introduce an algorithm, by which the sign mapping from $\mathcal{I}_1$ to $\mathcal{I}_2$ will be possible. By each step, the necessary explanations are immediately followed. Notice that since we have proved the one to one mapping of indices, we do not present the indices for ease of understanding. 

    We indicate the multiplicative variables in the setting $\mathcal{I}_1$ with $\sigma_i$s and in the setting $\mathcal{I}_2$ with $\sigma'_i$s. Based on a fixed choice of $\sigma_i$s, we choose the values of $\sigma'_i$s such that the symbol signs in corresponding queries match. The algorithm is as follows.
    
    {\textbf{Step 1.}} \textit{Choose the messages with randomly chosen labels $j_1, j_2, ..., j_i$. Compare the query containing these messages when $\mathcal{I}_2$, i.e. $q_1^{(2)}=\pm\sigma'_{j_2j_3...j_i}W_{j_1}\pm\sigma'_{j_1j_3...j_i}W_{j_2}\pm...\pm\sigma'_{j_1j_2...j_{i-1}}W_{j_i}$, to the query when $\mathcal{I}_1$, i.e. $q_1^{(1)}=\pm\sigma_{j_2j_3...j_i}W_{j_1}\pm\sigma_{j_1j_3...j_i}W_{j_2}\pm...\pm\sigma_{j_1j_2...j_{i-1}}W_{j_i}$. Simply choose $\sigma'_i$s in $q_1^{(2)}$ such that the sign of each symbol matches with the corresponding one in $q_1^{(1)}$.}

    {\textbf{Step 2.}} \textit{All the variables $\sigma'_i$ that were fixed in Step 1, appear also in some other queries, but not together. Go through all these queries, and fix other $\sigma'_i$s involved relative to the other already-fixed variable in Step 1.}

    Consider the query containing message labels $j_0, j_2, ..., j_i$, i.e. $q_2^{(2)}=\pm\sigma'_{j_2j_3...j_i}W_{j_0}\pm\sigma'_{j_0j_3...j_i}W_{j_2}\pm...\pm\sigma'_{j_0j_2...j_{i-1}}W_{j_i}$. The already-fixed variable $\sigma'_{j_2j_3...j_i}$ appears in this query too. Compare this query to its corresponding one when $\mathcal{I}_1$, i.e. $q_2^{(1)}=\pm\sigma_{j_2j_3...j_i}W_{j_0}\pm\sigma_{j_0j_3...j_i}W_{j_2}\pm...\pm\sigma_{j_0j_2...j_{i-1}}W_{j_i}$. Fix the other variables $\sigma'_{j_0j_3...j_i}, ..., \sigma'_{j_0j_2...j_{i-1}}$, relative to the already-fixed $\sigma'_{j_2j_3...j_i}$ such that either $q_2^{(2)}=q_2^{(1)}$ or $q_2^{(2)}=-q_2^{(1)}$.

    After fixing these queries (fixing the $\sigma'_i$s inside), one has the concern whether the fixed $\sigma'_i$s are consistent among the other queries they appear in simultaneously. For example, take the queries $q_1^{(2)}$ and $q_2^{(2)}$ fixed in Steps 1 and 2. In $q_1^{(2)}$ the message labels $j_1, j_2, ..., j_i$ and in $q_2^{(2)}$ the message labels $j_0, j_2, ..., j_i$ appear. In these two queries, the variables $\sigma'_{j_1j_3...j_i}$ and $\sigma'_{j_0j_3...j_i}$ are fixed. We should check in the query containing both of these together, i.e. containing message labels $j_0, j_1, j_3, ..., j_i$, whether their relative values remains consistent. In general we should prove, and this will also be needed in the following steps of our algorithm, whether any two variables of $\sigma'_i$s, when fixed in two different queries, maintain a correct relative value concerning in the query in which both of them appear. This will be proved in the following lemma. 

    \begin{lemma} \label{lem:signconsistency}
        The already-fixed variables $\sigma'_{j_1j_3...j_i}$ and $\sigma'_{j_0j_3...j_i}$, fixed in queries $q_1^{(2)}$ and $q_2^{(2)}$, maintain a correct relative sign when they appear together in another query. 
    \end{lemma}

    \begin{proof}
        We prove the lemma for one setting of the labels $j_0, j_1, j_2$ for each case of $\mathcal{I}_1$ and $\mathcal{I}_2$, since all other ones can be proved similarly. For ease of understanding, assume that every $\sigma_i=1$ when $\mathcal{I}_1$. Assume when $\mathcal{I}_1$, the messages with labels $j_0, j_1, j_2$ are all in the independent set with the ordering $j_0<j_1<j_2$. Additionally, assume among independent messages in $j_1, j_2, ..., j_i$, there is an odd number of messages between $j_1$ and $j_2$. Moreover among independent messages in $j_0, j_2, ..., j_i$, there is again an odd number of messages between $j_0$ and $j_2$. Based on this setting, after sign assignment we have the following queries for the three set of labels $\{j_1, j_2, ..., j_i\}$, $\{j_0, j_2, ..., j_i\}$, and $\{j_0, j_1, j_3, ..., j_i\}$ respectively,
        \begin{align}
            q_1^{(1)} = W_{j_1} + W_{j_2} \pm \cdots \\
            q_2^{(1)} = W_{j_0} + W_{j_2} \pm \cdots \\ 
            q_3^{(1)} = W_{j_0} - W_{j_1} \pm \cdots
        \end{align}

        For $\mathcal{I}_2$, we consider the case where labels $j_1$ and $j_2$ are among the independent messages which have odd number of independent messages in between based on the ordering among $j_1, j_2, ..., j_i$. Additionally we assume $j_0$ is among dependent messages. With this setting, if we assume the relative sign between $W_{j_0}$ and $W_{j_2}$ in $q_2^{(2)}$ is minus, then we have,
        \begin{align}
            q_1^{(2)} = \sigma'_{j_2j_3...j_i} W_{j_1} + \sigma'_{j_1j_3...j_i} W_{j_2} \pm \cdots \\
            q_2^{(2)} = \sigma'_{j_2j_3...j_i} W_{j_0} - \sigma'_{j_0j_3...j_i} W_{j_2} \pm \cdots \\ 
            q_3^{(2)} = \sigma'_{j_1j_3...j_i} W_{j_0} + \sigma'_{j_0j_3...j_i} W_{j_1} \pm \cdots
        \end{align}

        To fix the variables in $q_1^{(2)}$ and $q_2^{(2)}$, we should set $\sigma'_{j_2j_3...j_i}=\sigma'_{j_1j_3...j_i}=1$ and $\sigma'_{j_0j_3...j_i}=-1$. This leads to $q_3^{(2)} = W_{j_0} - W_{j_1} \pm ...$, which as can be seen, automatically matches with $q_3^{(1)}$. So the relative signs remain consistent and the lemma is proved. 
        
    \end{proof}

    \begin{remark}
        The reason why only the relative values of $\sigma'_i$s are important, is because of the sign switching variables of Step 5.2 in the scheme. When the relative signs of symbols are correct, to match these signs between two corresponding queries of different labelings $\mathcal{I}_1$ and $\mathcal{I}_2$, we only need to multiply the whole query with a $-1$ or a $+1$. 
    \end{remark}

    In Step 2, we fixed all the queries that are within $1$ message distance from the first randomly chosen query $q_1^{(2)}$; meaning the queries in Step 2 have $i-1$ messages in common with that of $q_1^{(2)}$ and are only different in $1$ message. In Step 3, we fix the queries with distance $2$ from $q_1^{(2)}$.
    
    {\textbf{Step 3.}} \textit{Consider all the queries with distance $2$ from $q_1^{(2)}$. Fix the variables $\sigma_i'$s within these queries relative to the already-fixed ones in the first two steps.}

    Consider the query containing messages with labels $j'_1, j'_2, j_3, ...,j_i$, which is in distance $2$ from $q_1^{(2)}$. The variables $\sigma'_{j'_2, j_3, ...,j_i}$ and $\sigma'_{j'_1, j_3, ...,j_i}$ have been already fixed in Step 2 of the algorithm, and they both appear in the mentioned query. We should prove their relative value remains correct in this new query. This is proved in the following lemma. 

    \begin{lemma}
        The already-fixed values of $\sigma'_i$s within queries in previous steps, maintain the correct relative values in Step 3.
    \end{lemma}

    \begin{proof}
        Take two queries containing the message labels $j'_1, j_2, ..., j_i$ and $j'_2, j_2, ..., j_i$. These queries are fixed in Step 2, so the values for $\sigma'_{j'_2, j_3, ...,j_i}$ and $\sigma'_{j'_1, j_3, ...,j_i}$, are already fixed in these two queries. Exactly like the proof in Lemma \ref{lem:signconsistency}, the relative signs of these variables remain correct in the query with labels $j'_1, j'_2, j_3, ...,j_i$, which contains both variables.
    \end{proof}

    The rest of the algorithm is evident. 

    {\textbf{Step 4.}} \textit{Each time increase the distance of queries from $q_1^{(2)}$ by one, and fix the not-yet-fixed $\sigma'_i$s within these queries. Continue this process until the last step, where the distance is $i$. Then, all the queries will be exhausted and fixed.}

    The correctness of relative signs of the already-fixed variables in each step is proved similar to the previous steps. 

    It is readily evident by our algorithm, that if the mapping of symbol signs from the setting $\mathcal{I}_2$ to $\mathcal{I}_1$ is done by the values $\{\sigma'^{*}_i\}$ and sign switching variables in vector $\mathbf{s}$, then there would be another set of answers $\{-\sigma'^{*}_i\}$ and $-\mathbf{s}$ and there exists no other set of answers. This proves that the mapping from all possible setting to the setting $\mathcal{I}_1$, is uniformly random, thus hiding the private labeling in Step 1 of the scheme. This completes the proof the privacy.

\section{Proofs of Theorem \ref{cor:repetition} and Theorem \ref{thm:order-optimal}} \label{sec:rep}
 	
 	To calculate the rate of the repetition scheme which uses the PC scheme for each demanded message separately, we note that at each use, the scheme downloads extra decodable symbols from other demanded messages. Thus, the rate for the repetition scheme $R_{1}$ would be 
 	\begin{align}
 		R_{1} = C + \Delta(M,K,P,N),
 	\end{align} 
 	where $C$ is the capacity of the single-message private computation scheme which is achieved by the PC scheme and $\Delta(M,K,P,N)$ is the rate due to extra decodable symbols in every use of PC. In round $1$ of PC, the user downloads one new symbol for each $K$ independent message (and because of the dependency involved, one new symbol for each of the $M$ messages). So in each use of PC for one demanded message, the user downloads extra $P-1$ symbols from the other $P-1$ demanded messages on each server. Therefore, the total number of extra symbols downloaded in each use of PC is $(P-1)N$. To calculate $\Delta(M,K,P,N)$, we proceed as follows. We note that the total download in the PC scheme follows $D_{PC} = \frac{L_{PC}}{C}$ and since $L_{PC}=N^M$ and $C=\frac{1-\frac{1}{N}}{1-(\frac{1}{N})^K}$, $D_{PC} = \frac{N^M \left(1-(\frac{1}{N})^K\right)}{1-\frac{1}{N}}$. Thus, we have
 	\begin{align}
 		\Delta(M,K,P,N) = \frac{(P-1)N}{D_{PC}} = \frac{(P-1)(N-1)}{N^M \left(1-(\frac{1}{N})^K\right)}.
 	\end{align}
 	
 	In terms of privacy, since every single use of PC is private, the repetition scheme would also be private. The rate of this scheme corresponds to the first term of the maximization in Theorem \refeq{cor:repetition}. 
 	
 	On the other hand, by treating each message as an independent one, we can use the MM-PIR scheme of~\cite{banawan2018multi} as a solution to the MM-PC problem. The rate of this scheme corresponds to the second term in the maximization. Therefore, Theorem \ref{cor:repetition} is proved.

    To prove Theorem \ref{thm:order-optimal}, we first note that the capacity of the MM-PIR defined in \cite{banawan2018multi} for $K$ total messages is  an upper bound to our problem, since this setting assumes independency among all messages and the MM-PC problem allows for requesting not only messages themselves, but also their linear combinations. For the case $P \leq \frac{K}{2}$, the upper bound for the MM-PIR problem would be $R^{\star}\leq R_u=\frac{1-\frac{1}{N}}{1-(\frac{1}{N})^{\lfloor{\frac{K}{P}}\rfloor}}$. 
        Note that the achieved rate in~\eqref{eq:rep scheme rate} is no less than $\frac{1-\frac{1}{N}}{1-(\frac{1}{N})^K}$, which is achieved by 
        using the PC scheme in~\cite{sun2018capacity} $P$ times. 
        Thus
        \begin{align}
            \frac{R_u}{R_1} \leq \frac{\frac{1-\frac{1}{N}}{1-(\frac{1}{N})^{\lfloor{\frac{K}{P}}\rfloor}}}{\frac{1-\frac{1}{N}}{1-(\frac{1}{N})^K}}=\frac{1-(\frac{1}{N})^K}{1-(\frac{1}{N})^{\lfloor{\frac{K}{P}}\rfloor}}\leq \frac{1}{1-\frac{1}{N}} \leq 2.
        \end{align}
        For the case $P \geq \frac{K}{2}$, the capacity of MM-PIR follows $R_u=\frac{1}{1+\frac{K-P}{PN}}$. Thus
        \begin{align}
            \frac{R_u}{R_1} \leq \frac{\frac{1}{1+\frac{K-P}{PN}}}{\frac{1-\frac{1}{N}}{1-(\frac{1}{N})^K}} \leq \frac{1}{1-\frac{1}{N}} \leq 2.
        \end{align}
  
 	        \section{Proof of Lemma~\ref{lem:redundancy}} 
        \label{sec:proof of redundancy}
We first point out that the structure of the queries in each stage, up until the end of Step 4 (index assignment), is exactly like the structure of the multicast messages in the delivery phase of the MAN coded caching scheme with $M$ files and $M$ users in which every user demands a different file; thus all the files are requested. To restate the index structure in Lemma \ref{lem: index}, take a stage in round $i$ and choose any $i-1$ messages. The set of queries with these messages have the same index for the other symbol involved in the query. This is the exact same structure as in the delivery phase of the MAN scheme when $t=i-1$, where each multicast message includes $t+1$ users. 

In~\cite{wan2021optimal}, the authors show that when some of the demanded files are linear combinations of the others, by carefully designing the signs of each symbol in the delivery phase, some of the multicast messages are linear combinations of the other ones, and thus redundant. In their paper, the users requesting independent messages are called \textit{leaders}, and the other ones \textit{non-leaders}. Therefore in our scheme, the independent messages correspond to the leaders, and the dependent ones to the non-leaders. In~\cite[Appendix B]{wan2021optimal} they show using the structure plus in sign assignment, the multicast messages which do not include any leaders, are redundant and can be derived by other multicast messages. This is the first part of the proof.


On the other hand, in a stage, we can take a slight modification on the composition of the multicast messages in~\cite[Eq. 54]{wan2021optimal}, where the sign between the required blocks by the leaders and the non-leaders is changed from $+1$ (structure plus) to $-1$ (structure minus) such that the new composition of  $X_{\mathcal{S}}$ becomes 
\begin{align}
  X_{\mathcal{S}}&=\sum_{i\in [|\mathcal{L}_{\mathcal{S}}|]} (-1)^{i-1} B_{\mathcal{L}_{\mathcal{S}}(i), \mathcal{S}\setminus \left\{\mathcal{L}_{\mathcal{S}}(i)\right\}} \nonumber\\&  - \sum_{j\in [|\mathcal{N}_{\mathcal{S}}|]}(-1)^{j-1} B_{\mathcal{N}_{\mathcal{S}}(j),\mathcal{S}\setminus \{\mathcal{N}_{\mathcal{S}}(j)\}}. \label{eq:general multicast message in Fq}  
\end{align} 
By the new multicast message composition in~\eqref{eq:general multicast message in Fq}, we can still prove the~\cite[Eq. (57a)]{wan2021optimal} holds, which refers to the redundancy of some multicast messages, but with sightly modified decoding coefficients 
\begin{align}
    \beta_{\mathcal{A},\mathcal{S}}= (-1)^{1 +\text{Tot}(\overline{\text{Ind}}_{\mathcal{S}})+ |\mathcal{S}\setminus \mathcal{A}|} \text{det}(\mathbb{D}^{\prime}_{\mathcal{A} \setminus \mathcal{S},\mathcal{L}_{\mathcal{S}} }). \label{eq:decoding coefficients}
\end{align}
The proof of~\cite[Eq. (57a)]{wan2021optimal} with new multicast message composition in~\eqref{eq:general multicast message in Fq} and decoding coefficients in~\eqref{eq:decoding coefficients} directly follows the same steps as in~\cite[Appendix B]{wan2021optimal}, and thus we do not repeat it. This proves the same redundancy exists with the structure minus of sign assignment. Notice that the sign switching variables in Step 5.2 clearly does not affect the redundancy. This completes the proof of the theorem.

		\bibliographystyle{IEEEtran}
 	\bibliography{references}
\end{document}